\documentclass[manuscript,screen]{acmart}
\usepackage{graphicx}
\usepackage{subcaption}
\newtheorem{fact}{Fact}
\newtheorem{remark}{Remark}

\AtBeginDocument{%
  \providecommand\BibTeX{{%
    \normalfont B\kern-0.5em{\scshape i\kern-0.25em b}\kern-0.8em\TeX}}}

\setcopyright{acmcopyright}
\copyrightyear{2018}
\acmYear{2018}
\acmDOI{XXXXXXX.XXXXXXX}

\acmJournal{JACM}
\acmVolume{37}
\acmNumber{4}
\acmArticle{111}
\acmMonth{8}

\acmPrice{15.00}
\acmISBN{978-1-4503-XXXX-X/18/06}

\begin{document}

\title{A Statistical Verification Method of Random Permutations for Hiding Countermeasure Against Side-Channel Attacks}

\author{Jong-Yeon Park}
\email{jonyeon.park@samsung.com}
\orcid{}
\affiliation{%
  \institution{Samsung Electronics S.LSI}
  \streetaddress{Samsung Electornics Road 1-1}
  \city{Hawsung}
  \state{Kyeung-Ki Province}
  \country{Korea. Rep}
  \postcode{}
}  
\author{Jang-Won Ju}
\email{jangwon95.ju}
\affiliation{%
  \institution{Samsung Electronics S.LSI}
  \streetaddress{Samsung Electornics Road 1-1}
  \city{Hawsung}
  \state{Kyeung-Ki Province}
  \country{Korea. Rep}
  \postcode{}
}

\author{Wonil Lee}
\email{wonil01.lee}
\affiliation{%
  \institution{Samsung Electronics S.LSI}
  \streetaddress{Samsung Electornics Road 1-1}
  \city{Hawsung}
  \state{Kyeung-Ki Province}
  \country{Korea. Rep}
  \postcode{}
}

\author{Bo-Gyeong Kang}
\email{bogyeong.kang@samsung.com}
\affiliation{%
  \institution{Samsung Electronics S.LSI}
  \streetaddress{Samsung Electornics Road 1-1}
  \city{Hawsung}
  \state{Kyeung-Ki Province}
  \country{Korea. Rep}
  \postcode{}
}

\author{Yasuyuki Kachi}
\email{kachi@u-aizu.ac.jp}
\affiliation{%
  \institution{University of Aizu}
  \country{Japan}
}

\author{Kouichi Sakurai}
\email{sakurai@inf.kyushu-u.ac.jp}
\affiliation{%
  \institution{Kyushu University}
  \country{Japan}
}

\renewcommand{\shortauthors}{Jong-Yeon Park, et al.}

\begin{abstract}
  As NIST is putting the final touches on the standardization of PQC (Post Quantum Cryptography) public key algorithms, it is a racing certainty that peskier cryptographic attacks undeterred by those new PQC algorithms will surface. Such a trend in turn will prompt more follow-up studies of attacks and countermeasures. As things stand, from the attackers’ perspective, one viable form of attack that can be implemented thereupon is the so-called “side-channel attack”. Two best-known countermeasures heralded to be durable against side-channel attacks are: \textit{“masking”} and \textit{“hiding”}. In that dichotomous picture, of particular note are successful single-trace attacks on some of the NIST’s PQC then-candidates, which worked to the detriment of the former: \textit{“masking”}. In this paper, we cast an eye over the latter: \textit{“hiding”}. Hiding proves to be durable against both side-channel attacks and another equally robust type of attacks called “fault injection attacks”, and hence is deemed an auspicious countermeasure to be implemented. Mathematically, the hiding method is fundamentally based on random permutations. There has been a cornucopia of studies on generating random permutations. However, those are not tied to implementation of the hiding method. In this paper, we propose a reliable and efficient verification of permutation implementation, through employing Fisher–Yates’ shuffling method. We introduce the concept of an $n$-th order permutation and explain how it can be used to verify that our implementation is more efficient than its previous-gen counterparts for hiding countermeasures.
\end{abstract}



\keywords{Side Channel Attack, Hiding method, Random Permutation, Post-Quantum Cryptography}


\maketitle

\section{Introduction}
\subsection{Background}
The advent of quantum computing is sweepingly changing the game of cryptography. It is taken as gospel truth that quantum computing will soon evolve to the level it undermines one’s confidence in the pre-existing—once thought-to-be inveterate—encryption algorithms such as RSA and Elliptic curve cryptography (ECC), which have been around since the dawn of the digital communication age. Shor’s research outright favors such a prognosis: Shor brought to light the vulnerability of RSA and ECC against quantum computing~\cite{SOR94}. To cope with such development, in 2016, the National Institute of Standards and Technology (NIST) initiated the standardization of Post-Quantum Cryptography (PQC). After multiple sequential rounds of competition protracted over 6+ years, four groups of algorithm builders that built a key encapsulation mechanism (KEM) and digital signature schemes ultimately secured a win, which was as recently as August 2023~\cite{NIST20}.

As the crypto niche’s own market force (practically NIST) laid down the law to allow only a handful of algorithms to ride through the waves and outlive the contenders, attacks against those “last algorithms standing” are naturally imminent. Among them, side-channel attacks are well-deserving of scrutiny: Side-channel attacks do not target the algorithms directly, but rather go after secondary, collateral, information laid bare through reconnoitering. The best-known countermeasures against side-channel attacks are \textit{“masking”} and \textit{“hiding”}. (We give a brief pr\'ecis of masking and hiding in section~\ref{pre}.) Most notably, the hiding technique employs random permutations of independent unit operations, making it challenging for an attacker to pinpoint the location and the timing of the specific operation being performed.

One noticeable caveat of the NIST-vouchsafed PQC algorithms is the immensity of their polynomial operation processing. Indeed, those algorithms fundamentally rely on the supposed impenetrable nature of some well-chosen lattice problems, where polynomial operations serve as the main executing method. In that setup, the number of unit data blocks grows in proportion to the order of the polynomials processed during the cryptographic operations~\cite{INDO_shuf}. As an example of comparative assessment, while AES operates on 16 blocks, PQC algorithms process more than 32 blocks. Meanwhile, some other PQC algorithms that are not based on lattice problems, such as the code-based cipher HQC~\cite{NIST22}, use vector and matrix operations with independently operated variables. The hiding method is known to mesh well with all these algorithms.  

Given that the NIST-vouchsafed PQC algorithms operate on a large number of unit blocks as compared to their predecessors, it is imperative to verify the process how source code developers for software and hardware designers arrange these blocks when designing the hiding technique. That is crucial as these methods have bearings on both the security and the efficiency of the systems.

In light of the foregoing, in this paper we report the findings on the hiding technique. What fundamentally underpins the hiding method is a mathematical principle called \textit{random permutation}: whenever a function is executed over a portion of its domain, the order of the independent parts is randomized to the greatest extent possible.

\subsection{Motivation}
Our research was provoked by the successful side-channel attacks against one of the NIST’s then-PQC algorithm candidates using only a single trace~\cite{single_pqc,single_KMU}. The well-executed cryptanalytic feat has revealed that the masking technique is ineffective against side-channel attacks. That and the following evidence would give a lift to the prevailing nature of the hiding technique: Fault injection–a well-known kissing cousin of side-channel attacks–can be thwarted using the hiding technique~\cite{SE_PQC}.

On a more mathematical side, we observed that the complexity of \textit{shuffling} is an overriding element that determines the security strength of countermeasures against a certain type of attacks including single trace attacks. Therefore, verifying the complexity of \textit{permutations} that are generated on-the-spot is necessary to ensure the security of the implementation.

Of particular interest is the fact that timing attacks (which make up a subgenre of side-channel attacks) targeting code-based algorithms including NIST’s Round4 HQC have been identified~\cite{timing_HQC}. That study instigated the development of reference codes for key-generation algorithms that incorporate permutations as a countermeasure. This demonstrably suggests that secure implementation of permutations significantly affects the security level of the algorithm’s implementation~\cite{NIST22,secure_BIKE}. 

We now got all the bases covered. The above narrative paves the way for the following: 

Fisher–Yates (F-Y) shuffling is a well-known algorithm for uniform random permutation. It is widely used in the hiding countermeasures~\cite{FY}. If meticulously and dexterously implemented, the same algorithm can be a generator of perfect permutations. Moreover, provided the main operations to be hidden are in store, one is not compelled to reduce the level of complexity, especially in a software environment, since a software code must be operated sequentially. From it one can reasonably infer that F-Y shuffling must be a strongly favored algorithm to implement for a hiding operation. There are at least three major factors that influence the randomness of the implemented permutation–this observation is indeed the \textit{raison d’être} of this paper: (i) Incorrect implementation. (ii) Poor random source. (iii) Misapplication of random numbers (the chosen random number being too small). Suppose that a certain implementation method produces biased results. In that instance, it is pivotal to have a cut-and-dried formula at our disposal that discerns whether the implementation was outright flawed.

\subsection{Related Works and Challenges}

Pre-existing studies with regards to the validation of the uniformity of random permutations are based on the methods as described in~\cite{per_a1,per_a2}.  Those studies gear more towards the theoretical uniformity of the algorithm as opposed to \textit{bona fide} implementation. They do not address the randomness of permutations.

If the number of objects being shuffled is too large (in our example, $N\geq 32$, say), then verifying its uniform distribution becomes extremely challenging due to the fact that there are an exorbitant number of cases to examine ($N!$ to be exact). 

V. Charvillon et.al.~\cite{shuffling} conducted an extensive study on the hiding countermeasure. They proposed high-speed permutation generators and concluded that the hiding countermeasure is not plagued by a small bias precipitated from their permutation generation algorithm in the event of side-channel attacks. In short, small biases do not boost the capability of side-channel attacks, thus they are deemed negligible. They described how they assessed the bias of the permutation generation, namely: they computed the mutual Euclidean distances between the uniform random distributions for the values of N with $3\leq N\leq 9$. Due to the brute force nature of this approach, verification when $N\geq 10$ could not be extrapolated, at least within a manageable sample space size (and a manageable amount of time). 

R. Mitchell et. al.~\cite{Shu_chi} suggested a generator of random permutations using a GPU. They validated the randomness of its permutation generation with a $\chi^2$ test, albeit with a caveat: they only tested it for $N$ with $N\leq 5$. When $N$ is somewhere around 32 (in which case we say $N$ is in \textit{“normal size”}), an astronomical number of tests is required, which is virtually impossible to carry out in a reasonable amount of time.

All the while, there is a consummate body of mathematical research on uniform and non-uniform permutations~\cite{ref_per1,ref_per2,ref_per3,ref_per4,ref_per5,Phd_T}. (Note: some are as old as the 1960s and some are as recent as 2016–.) However, those are too theoretical to be applied as a ``distinguisher'' ({\it i.e.\/}, a criterion that demarcates the boundary between the uniformity and the non-uniformity). Practical tools to easily verify the uniformity of a small sample space are called for.

\subsection{Our Contribution and the novelty of research} 
The gist of our study boils down to figuring out how best to conduct the experiment. As a starter, we furnish one well-defined permutation, whose order is $n-1$, say. With it, we conduct experiments on various F-Y shuffling implementation methods, and discern which one of the methods works best. Then we use the results to determine which implementation method is the most secure.

We prove that, an $n$-th order permutation is an $(n-1)$-th order permutation. Furthermore, we prove that, if the chosen permutation has order $N-1$, where $N$ coincides with the number of nodes in the permutation, then perfect random security is achieved. Note that the two notions, the notion of an $N$-th order permutation, and the widely accepted notion of an $N$-th order side-channel attack, are in mutually redeeming relations.

The contributions of this paper are as follows:

\begin{itemize}
\item We furnish four different kinds of implementations of F-Y shuffling, and pit them against each other: two standard implementations and two poorly structured implementations.

\item  We propose a new, viable verification method for random permutations, namely, the $(N-1)$-th order permutation estimator. We also provide the theoretical background and elucidate the reason why we chose this method. We contend that this is a viable estimation tool for side-channel attack evaluators.

\item  We compare our verification method against the brute-force factorial complexity estimation, and we substantiate our method’s efficacy. We devise a way to expeditiously reduce the number of cases in examining a large pool of permutations through experimentation. 

\item  We present experimental results that corroborate with the conclusion we drew, namely, the following three factors influence (work for or against) the uniformity of the F-Y shuffling operation: (i) Incorrect implementation. (ii) Poor random source. (iii) Misapplication of random numbers. This result should be deemed instrumental for developers of side-channel attack countermeasures. To the best of our knowledge, what this paper expounds is the first study ever to evaluate the real-world adverse impact of incorrect implementations of F-Y shuffling caused/aggravated by the aforementioned three factors (i--iii).

\end{itemize}

\section{Prerequisites}\label{pre}
\subsection{Side channel attack and hiding countermeasure}

Paul Kocher was the first to propose the concept of side-channel analysis~\cite{timming}: it has since kept garnering a wealth of ideas and theories, and today it is in a greatly ameliorated shape. The common-denominator philosophy of cryptographic attacks banks on the premise that sensitive data are calculated (data-crunched) or otherwise unwittingly exposed in a specific time-period. Attacks epitomizing that philosophy resort to statistical analysis of multiple power and electromagnetic (EM) traces~\cite{CPA,template,DPA}. They attempt to ferret out the key by analyzing a single trace~\cite{SPA}. In recent years, single/multiple trace attacks have evolved into many different forms.

Archetypal examples include side-channel collision attacks~\cite{collision}, soft analytical side-channel attacks (SASCA)~\cite{SASCA}, and single trace attacks on PQC ~\cite{single_pqc,single_KMU}. A once-favored countermeasure against a multi trace analysis was the masking technique~\cite{book}. The method generates a masking value in advance, and recalibrates sensitive operations so as to prevent an unwitting exposure of the original data. As masking distorts the original value, a restoration operation is required to recover the masked value. Therefore, implementing the entire countermeasure securely through the masking method is a rather challenging undertaking, since it is a non-linear operation, and performance failure can occur at the configuration level. Moreover, the masking technique is not an effective countermeasure for a single trace analysis, indeed, it has been brought to light that the masking value itself can be discovered in attacks linked to a single trace analysis of the masking~\cite{k_trace}.

And here comes the alternative with a revamped feature—the hiding countermeasure. The hiding countermeasure involves shuffling of the order of operations that are mutually independent so as to prevent a specific operation from being caught sight of at a certain time~\cite{book,shuffling}. For instance, in the case of an AES S-box, which consists of parallel and completely independent operations, 16 S-boxes can be shuffled and calculated for each round. If shuffling is applied to 16 S-boxes independently, then the probability for a specific operation to be executed at a given location is $1/16$. 

Depending on the noise characteristics of the device, the probability of $1/16$ is deemed sufficient as a countermeasure against a general power (or EM) analysis. However, when defending against a more potent attack, such as a “second order attack”, we need the location information for two S-boxes to be shuffled with higher complexity in order to toughen the attack complexity. If the locations of the two S-boxes are mutually independent, then the position of the operation being executed is correctly pinned down with the probability of $1/ (15\times 16)$.

Assuming that all of the 16 keys are found with a single trace attack, the attack complexity being $15\times 16$ is too small. That is why we must resort to the notion of a so-called “perfect permutation”. If a perfect permutation is implemented, then an attacker would need to perform brute-force attacks on 16! cases (where $16!$ amounts to about 45 bits). Although 45 bits would not be characterized as a practically impossible level of attack complexity, once one incorporates another element of fortification, namely, the practical difficulty of a single trace analysis, one may safely conclude that the attack is already devilishly challenging.

\subsection{Fault injection attack and hiding countermeasure}

\textit{A fault injection attack}, a kissing cousin of a side channel attack, is a form of an active attack perpetrated by an attacker who deliberately introduces an error so as to ferret out the key value they are searching~\cite{bellcore}. By nature, one anticipates countless different attack scenarios, and that makes counteracting decidedly more challenging. For example, in a real-world safe-error attack~\cite{SE_PQC,SE_O}, an attacker cunningly identifies dummy operations, artfully alters the key values in memory to $0$ through fault injections, and that way successfully taps into the non-zero values associated with the keys. There are also numerous newer known attack variations that are proved to outwit the existing countermeasures. All in all, fault injection is–at least in theory–extraordinarily potent. 

Hiding operation proves to be thus-far the sole known, proved-to-be-effective, countermeasure against fault injection attacks. Indeed, the error variables unleashed by an attacker are tracked down and embargoed on the spot where the fault was inserted.

\subsection{Post Quantum Cryptography and hiding countermeasure}

Historically, hiding countermeasures were not a popular choice for protection of the public key encryption algorithms such as RSA and ECC, due to operational characteristics such as integer operation carryovers. Hiding countermeasures were found to be serviceable in {\it symmetric\/} key cryptography algorithms such as AES. In addition, the hiding method was also employed as an auxiliary measure for {\it asymmetric\/} key algorithms: indeed, the hiding method proved to have the redeemable quality of bolstering masking countermeasures and making up for the shortcomings of masking-only strategies~\cite{ACNS2007}.

Back to the NIST’s standardization process for Post-Quantum Cryptography (PQC) (mentioned in I. Introduction): after the third-round selection, the first cohort of competition winners was announced~\cite{NIST20}. Mathematically speaking, those NIST-vouchsafed algorithms are embodiments of lattice-, code-, and hash-based problems. Those cryptographically pivotal mathematical problems involve operations on polynomials, vectors, and matrices. Solving them requires the so-called “parallel programming processing”. They involve a sufficiently large number of “nodes”. Now, even with that sophisticated level of deterrence, they can still be vulnerable if equipped only with masking countermeasures~\cite{BP_PQC}. 

It is for this reason that hiding countermeasures are in demand as effective countermeasures for PQC algorithms. When the number of independent operations exceeds a certain threshold, convincing results as to the effectiveness of the hiding counternmeasures are naturally expected. That fact affords us to do away with the masking technique. While the effectiveness also depends on the noise level and the extent of power leakage, it is to all appearances viable to implement hiding countermeasures without incorporating masking countermeasures. In a nutshell, in the era of PQC, one’s ability to select a perfect permutation—the pivotal mathematical notion that underpins the hiding techniques—is likely a game-changer for the security durability of an algorithm. The fact that the hiding countermeasures are no longer relegated to a second place behind masking techniques strongly supports that narrative.

\subsection{Basic mathematical requisite for hiding countermeasure : Permutation}

In mathematics, the term ``permutation'' refers to a one-to-one map from a set onto itself, or the same to say, a bijective self-map. When the term ``permutation'' is used, often (though not always) it is understood that the referenced set is assumed to be finite. For each positive integer $N$, $\mathbb{Z}_N$ denotes the set of integers $j$ such that $0 \leq j \leq N-1$: $\mathbb{Z}_N = \{0, 1, \cdot \cdot \cdot , N-1\}$. Any (non-empty) finite set is in one-to-one correspondence with $\mathbb{Z}_N$ for some $N$.

{\begin{definition}\emph{$($\textup{permutations}$)$} \textup{A permutation is a \textit{bijective} self-map $\mathbb{Z}_N \longrightarrow \mathbb{Z}_N$.} 
\end{definition}

Technically, this definition itself does neither involve integer arithmetic nor orders (``$\leq$'', \textit{etc.}). Thus we may instead furnish any (abstract) finite set $X$, and call a bijective self-map $X \longrightarrow X$ a permutation. However, we prefer $\Bbb Z_N$ over an abstract set $X$ due to the nature of the ensuing discussion. Let $S$ be the set of all permutations $\Bbb Z_N \longrightarrow \Bbb Z_N$: $S = \{f_1, \cdot \cdot \cdot , f_M\}$, where $M$ is the cardinality of the set $S$: $M = |S|$. It is well-known that $M$ equals the factorial of $N$: $M = N!$.

{\begin{definition}\emph{$($\textup{select functions}$)$} \textup{Let $S = \{f_1, \cdot \cdot \cdot , f_M\}$ $(M = N!)$ be as above. Let $T$ be a positive integer. A \textit{select function} is a map $\gamma : \Bbb Z_T \longrightarrow S$.} 
\end{definition}

Note that, in Definition 2, we do neither assume that $\gamma$ is injective nor surjective. Also note that $T$ is an arbitrary positive integer. In particular, we do not preclude the case $T < M (= N!)$ from consideration.

It is convenient to view an element of $S$ as a \textit{shuffling} of a deck of $N$ total number of cards (labeled $1$, $2$, $\cdot \! \cdot \! \cdot$, $N $). That in turn prompts us to view $\gamma : \Bbb Z_T \longrightarrow S$ as a finite sequence of shuffling. Namely, $\gamma$ takes each $t \in \Bbb Z_T = \{0, 1, \cdot \cdot \cdot , T-1\}$ into some shuffling of a deck: $t = 0$ designates one shuffling of that deck; $t = 1$ designates another shuffling of the same deck (it can be the same shuffling as the $t = 0$ case); $t = 2$ designates yet another shuffling of the same deck (it can be the same shuffling as the $t = 0$ or the $t = 1$ case), and so forth. Viewing it from another angle: when one shuffles that deck $T$ consecutive number of times, that yields $\gamma : \Bbb Z_T \longrightarrow S$. With that mathematical scope in mind, let us agree that AES is simply a select function whose argument is a key.

Towards our goal of making a cryptographic algorithm an unadulterated secure commodity using only the hiding technique, it is crucial to engineer a permutation generator that generates permutations that bear the satisfactory level of complexity (``\textit{safe permutation}''). To that end, we theoretically offer a criterion for a safe permutation, and experimentally test that theory. We must have a way of knowing whether each permutation designated by a select function $\gamma : \Bbb Z_T \longrightarrow S$ forms a \textit{uniform distribution} (= \textit{perfect random security} in Definition 3). The hurdle to overcome resides in the fact that, when the number $N$ grows super-exponentially (\textit{e}.\textit{g}., at the growth rate of the factorial sequence $\{0!, 1!, 2!, \cdot \cdot \cdot , n!, \cdot \cdot \cdot\}$), determining whether a given permutation $f_j : \Bbb Z_N \longrightarrow \Bbb Z_N$ is a uniform distribution through testing a large case sample space is practically impossible. In this paper, we introduce the notion of a so-called ``security estimator'', and furthermore propose a viable way to verify the safeness (``\textit{completeness}'') of the given permutations. Note that, when the security estimator has order $N-1$, it appears to be complex enough to defy any meaningful testing, however, in reality, that level of complexity is negligible as compared to the factorial complexity.

In this paper, we re-calibrate the Fisher-Yates shuffling~\cite{Knuth,FY} as a select function. We consummately perform all necessary calculations. We establish the fact that, theoretically, Fisher-Yates shuffling can produce an arbitrary uniform permutation on a given finite set.

\section{New evaluation framework for permutations}

In this section we introduce a new evaluation criterion for permutations. The result is summarized in Theorem~\ref{coro}: if a select function $\gamma : \Bbb Z_T \longrightarrow S$, where $|S| = N!$, has order $N \! - \! 1$, then $\gamma$ possesses perfect random security. We substantiate our claim that our verification method (sufficient condition) when (under what circumstance) a given select function has order $N \! - \! 1$ is applicable to decide whether a given select function is a “completely uniform” permutation (= \textit{perfect random security}). To that end, we introduce several concepts.

\subsection{Perfect random security and k-th order permutation}

First, we introduce the concept ``{\it perfect random security\/}''  (Definition~\ref{perfect} below). This definition encapsulates the idea that, some particularly ``well-trimmed'' select function $\gamma$  \textit{selects all permutations with equal probability}. In theory we can refer to Definition~\ref{perfect} in deciding whether the select function achieves the perfect random security. If it indeed does, then it means that the implementation is of the most ideal form.

\begin{definition}\label{perfect}
\emph{$($\textup{perfect random security}$)$} 
\textup{A select function $\gamma : \Bbb Z_T \rightarrow S$ (where $S$ is the entire set of permutations of $N$ indices, in particular, $|S| = N!$) is said to have {\it the perfect random security\/}, if for an arbitrary $f \in S$, the probability $P(\gamma(m) = f) = 1/N!$, where $m \! \in \! \mathbb{Z}_T$ is randomly chosen. Or, equivalently: A select function $\gamma$ is said to have {\it the perfect random security\/}, if $\gamma$ is surjective and moreover, $\# \gamma^{-1} (f)$ (the number of elements in the fiber $\gamma^{-1} (f)$) does not depend on the choice of $f \in S$.}
\end{definition}

The caveat of Definition 3 is that it is not easy to decide whether a given select function $\gamma$ has the perfect random security, since the task boils down to an {\it exact\/} counting of the number of elements of a set. {\it Exact\/} counting is known to be very hard when the set is large. In our case, a large sample space is involved. Accordingly, below we propose an alternative concept with a redeeming feature, called a $k$-th order select function:

\begin{definition}\label{main_definition}
\emph{($k$$\textup{-th order select function}$ and $k$-th order permutations)} 

\textup{Let $k$ and $N$ be integers, $1 \leq k \leq N \! - \! 1$. A select function $\gamma : \Bbb Z_T \longrightarrow S$ (where $S$ is the entire set of permutations of $N$ indices, in particular, $|S| = N!$) is said to be of \textit{order} $k$ if, for an arbitrary pair of ordered $k$-tuples $(a_1, \cdot \! \cdot \! \cdot, a_k)$, $(b_1, \cdot \! \cdot \! \cdot, b_k)$, where 
$\{a_1, \cdot \! \cdot \! \cdot, a_k\}$ and $\{b_1, \cdot \! \cdot \! \cdot, b_k\}$ are subsets of $\{1, \cdots, N\}$ such that $\# \{a_1, \cdots, a_k\} = k = \# \{b_1, \cdots, b_k\}$,}

\vskip -5mm
//

$$
P\Big((f_m (a_k) = b_k) \Big| \bigcap_{i=1}^{k-1}(f_m(a_i) = b_i)\Big)=\frac{1}{\, N \! - \! k \! + \! 1 \,}
$$

\vskip -2mm

\noindent
\textup{holds, where $f_m = \gamma (m)$, and $m \in \Bbb Z_T$ is randomly chosen. Moreover, $\operatorname{Im} \gamma = \{f_m \, | \, m \in \Bbb Z_T\}$, where $\gamma$ is of order $k$, is referred to as a $k$\textit{-th order set of permutations}, or simply, $k$\textit{-th order permutations}.} 
\end{definition}

{Definition~\ref{main_definition} encapsulates a mathematical idea that best illustrates the so-called ``\textit{hiding countermeasure} \textit{against} \textit{higher-order attacks}'', a ``tiered'' version of the well-known side channel attacks. 

Below we make some purely mathematical observation pertinent to Definition 4, which also serves to \textit{demystify} the convolutedness of the definition stemming from the fact that it uses \textit{conditional} probability. To those readers who have the proclivity towards pure math, the following is just what the doctor ordered.

First, 1st order permutations are characterized by the equation
\vskip -2.5mm

$$
P\big((f_m(a_1) = b_1) \big) =\frac{1}{\, N \,}.
$$

\noindent
We may paraphrase it without resorting to the notion of probability is as follows: A select function $\gamma : \Bbb Z_T \longrightarrow S_N$ (where $S_N$ denotes the entire set of permutations of $N$ indices) is of order $1$ precisely when each of the sets $\{f_m (1) | m \in \Bbb Z_T\}$, $\{f_m (2) | m \in \Bbb Z_T\}$, $\cdots$, $\{f_m (N) | m \in \Bbb Z_T\}$ equals the entire $\{1, \cdots, N\}$, and moreover, $\# \gamma^{-1} (f)$ (the number of elements in the fiber $\gamma^{-1} (f)$) does not depend on the choice of $f \in \operatorname{Im} \gamma$.

Next, 2nd order permutations are characterized by the equation
$$
P\big((f_m(a_2) = b_2) | (f_m(a_1)= b_1)\big) =\frac{1}{\, N \! - \! 1 \,}.
$$

Paraphrasing of it without resorting to the notion of probability is as follows: We first define an auxiliary object $\rho_2$ (a map). For a permutation $f$ of $N$ indices $\{1, \cdots, N\}$ and for an ordered pair of indices $(a, b)$ where $a, b \in \{1, \cdots$, $N\}$, $a \not= b$, we have a new pair $(f(a), f(b))$, where $f(a)$ and $f(b)$ satisfy $f(a), f(b) \in \{1, \cdots, N\}$, and $f (a) \not= f(b)$. This way we obtain a permutation 
(a bijection) $f^{(2)}$ over the set $\{1, \cdot \cdot \cdot, N\} \times \{1, \cdot \cdot \cdot, N\} \backslash \Delta$, where $\Delta = \{(a, a) | 1 \leq a \leq N\}$. Lexicographical ordering of elements of $\{1, \cdot \cdot \cdot, N\} \times \{1, \cdot \cdot \cdot, N\} \backslash \Delta$ allows one to view $f^{(2)}$ as a permutation of $N (N \! - \! 1)$ indices. This way we arrive at a map $\rho_2 : S_N \longrightarrow S_{N (N \! - \! 1)}$; $\rho_2 (f) = f^{(2)}$.

\begin{fact}\label{product representation - 2}
Let $\gamma : \Bbb Z_T \longrightarrow S_N$ be a select function, where $S_N$ is the entire set of permutations with $N$ indices. Let $\rho : S_N \longrightarrow S_{N (N \! - \! 1)}$ be as above. Then $\gamma$ has order $2$ if and only if the composite $\rho_2 \circ \gamma : \Bbb Z_T \longrightarrow S_{N (N \! - \! 1)}$ has order $1$.
\end{fact}

Likewise, the definition of $k$-th order permutations can be paraphrased without resorting to the notion of probability, through $f^{(k)}$; $f^{(k)} (a_1, \cdots, a_k) = (f (a_1), \cdots, f(a_k))$ (where $a_1$, $\cdots$, $a_k$ are distinct elements of $\{1, \cdots , N\}$), and $\rho_k : S_N \longrightarrow S_{N (N \! - \! 1) \cdot \, \cdots \, \cdot (N \! - \! k \! + \! 1)}$; $\rho_k (f) = f^{(k)}$.

\begin{fact}\label{product representation - k}
Let $\gamma : \Bbb Z_T \longrightarrow S_N$ be a select function, where $S_N$ is the entire set of permutations with $N$ indices. Let $\rho : S_N \longrightarrow S_{N (N \! - \! 1)}$ be as above. Then $\gamma$ has order $k$ if and only if the composite $\rho_k \circ \gamma : \Bbb Z_T \longrightarrow S_{N (N \! - \! 1) \cdot \, \cdots \, \cdot (N \! - \! k \! + \! 1)}$ has order $1$.
\end{fact}

Now back to order $2$. In a 32-index permutation ({\it i.e.\/}, $N=32$), say, $P(f(4)=25|f(13)=7)=1/31$. Likewise, if we replace $4$, $25$, $13$ and $7$ with $a$, $b$, $c$ and $d$, where $a \not= c$ and $b \not= d$, then we still have $P (f (a) = b | f (c) = d) = 1/31$.} If a selection function $\gamma$ maintains perfect random security, then $\gamma$ trivially yields second-order permutations. To construct countermeasures against $k$-th order side channel attacks, we can simply use $k$-th order permutations (see Proposition~\ref{CPA_countermeausre}).

A set of second-order permutations is a set of first-order permutations. 
More generally, a set of $k$-th order permutations is a set of $(k-1)$-th permutations. We will prove this Theorem in the next section. 

To better understand Definition~\ref{main_definition}, consider a relatively simple example using permutations of four indices. Let \( S_4 \) denote the set of all permutations of \( \{1,2,3,4\} \). Assume that the select function \( \gamma \) maps onto the following subset of \( S_4 \):
\[
\{[1,2,3,4], [2,3,4,1], [3,4,1,2], [4,1,2,3]\}.
\]
Furthermore, assume that \( \gamma \) maps onto that set ``uniformly'', meaning that the number of elements in the fiber $\gamma^{-1} ([i_1,i_2,i_3,i_4])$ does not depend on the choice of $[i_1,i_2,i_3,i_4]$ in the set
\[
\{[1,2,3,4], [2,3,4,1], [3,4,1,2], [4,1,2,3]\}.
\]
Then \( \gamma \) is first order. However, according to the definition, \( \gamma \) is not second order. Similarly, even when a given select function $\gamma$ yields second order permutations, it may not yield third order permutations. This will be demonstrated in Appendix~\ref{App:Example_2nd} using an eye-catching example, a $40$-tuple of permutations with five indices possessing some curious symmetry. 

\begin{fact}\label{CPA_countermeausre}
As a countermeasure against the $k$-th order power analysis attack with $k$ different points of interest which is shuffled, $k$-th order permutation is enough to hide the operations.   
\end{fact}
Fact~\ref{CPA_countermeausre} demonstrates a bare minimum requirement for permutations to work as a hiding countermeasure for Shuffling, as dependent on the order of the side channel attacks. This fact is plausible (intuitively acceptable), and the reason is as follows. If an attacker targets two distinct shuffled positions out of \(N\) indices, the minimum probability that one attack point matches the two desired targets is given by 
\[
\frac{1}{N} \times \frac{1}{N-1}. 
\]
This corresponds to the 2nd order permutation in Definition~\ref{main_definition}. In the same manner, the $n$-th order permutation can also be similarly extended.

\subsection{Properties of a $(k-1)$th order select function}\label{permutation_property}

In this section, we examine the relationship between the notion of $(k-1)$th order and the notion of $k$-th order. We also establish the fact that the notion of $(N-1)$-th order and the notion of perfect random security are equivalent.

\begin{theorem}\label{coro}
If a select function $\gamma : \Bbb Z_T \longrightarrow S$, where $|S| = N!$, has order $(N \! - \! 1)$, then $\gamma$ has the perfect random security. 
\end{theorem}
\begin{proof}
The proof boils down to the well-known mathematical fact that any permutation of the set $\Bbb Z_T = \{0, \cdots , N-1\}$ is a product of transpositions.  
\end{proof}

Theorem~\ref{coro} is very useful. Indeed, in light of Definition~\ref{main_definition}, in order to verify that a select function $\gamma : \Bbb Z_T \longrightarrow S$, where $|S| = 32!$, has a perfect random security, it suffices to check that the probability of $f_m (a) = b$ is $1/2$, when $(x_{1}, y_{1})$, $\cdots$, $(x_{30}, y_{30})$ is an arbitrarily chosen, fixed, sequence} of 30 (out of 32) mutually distinct inputs $x_j$ and 30 outputs $y_j = f_m (x_j)$ $(j = 1, \cdots, 30)$, and when each of $a$ and $b$ is chosen randomly from the remainder: $a \in \Bbb Z_{32} \backslash \{x_j | j = 1, \cdots, 30\}$; $b \in \Bbb Z_{32} \backslash \{y_j | j=1, \cdots, 30\}$. 

Now let us move on to the next phase: we validate the experimental results statistically, with the expectation of a $1/2$ outcome.

\begin{remark}\label{main_theorem}  
If a select function $\gamma$ is of the $k$-th order for all $k\leq N$, then $\gamma$ is of the $(k-1)$th order. 
\end{remark}
\begin{proof}
The proof is in Appendix~\ref{proof_theorem}. 
\end{proof}

From Remark~\ref{main_theorem}, it follows that, if $\gamma$ supports a $k$-th order permutation, then it also supports a $(k-r)$-th order permutation, for each $r \geq 1$. Therefore, verifying that $\gamma$ suppports a 2nd order permutation is sufficient to protect against both 1st and 2nd order attacks.

\subsection{The determination of security }\label{SD}

By the definition of $(N-1)$th order permutation, the validation of the property is to check probability $P(f(x)=y|\delta)=1/2$ with a condition $\delta$. Because there are only two cases $f(x)=y_0$ and $f(x)=y_1$ with the condition $\delta$, we count all cases satisfying $f(x)=y_0$ and $f(x)=y_1$. If the permutation is uniform, we expect $P(f(x)=y_0|\delta)=1/2$. Considering this situation, it is Bernoulli trials where ${ X }_{ 1 }{ ,X }_{ 2 },...,{ X }_{ n }$ are independent. Their sum is distributed according to a binomial distribution($\mathcal{B}$) with $n$ and $\frac{1}{2}$

\[\Psi=\sum^n_{k=1}{X_k}\sim \mathcal{B}(n,\ \frac{1}{2})\] 

If $n$ is sufficiently large, $\mathcal{B}(n,\ \frac{1}{2})$ is given a normal distribution($\mathcal{N}(\mu,\sigma^{2})$) as 

\[\Psi\sim\mathcal{N}(n/2,\ n/4)\] 

For example, if $n$ equals $3000$ which is sufficiently large, the binomial distribution follows the normal distribution $\mathcal{N}(1500, 750)$, and we can determine the distribution is biased or not, using a probability density function, where the lower cumulative distribution $\mathcal{P}\left(x,1500,\sqrt{750}\right)$ is computed as 

\[\int^x_{-\infty }{\frac{1}{\sqrt{2\pi \sigma }}e^{-\frac{1}{2}{\left(\frac{x-1500}{\sigma }\right)}^2}} d \sigma.\]

If the number of cases satisfying $f(x)=y_0$ is $1600$ in $3000$ trials, $\mathcal{P}\left(1600,\ 1500,\sqrt{750}\right)=0.999869$. Here, the case is under probability $1 - 0.999869 = 0.000131$. \\

\subsubsection{Ratio comparison}The probability $0.000131$ is very low. However, we cannot determine how it is biased, because there are so many cases to be estimated in $(N-1)$th order permutation criteria. Thus, we must check the number of cases which is compared with the expected probability computed by lower cumulative distribution, and the rate of full cases. For instance, if the number of cases is $n=10^5$, then the expected number($EN$) of appearance cases is $0.000131\times n = 13.1$. This means that we can anticipate roughly 13 cases to exceed 1600 in 3000 trials under normal circumstances. Thus, we can compare the ratio between observed number and expected normal number. We call it a ratio comparison. \\

\subsubsection{Chi Square($\chi^2$) Test}
The probability variables following the normal distribution can be changed to Standard normal distribution $z=(\xi-\mu)/\sigma$. 
From the fact that each result is in distribution $\Psi\sim \mathcal{N}(n/2,\ n/4)$, we can make the random variables follow $\mathcal{N}(0,1)$.
If $\Psi_1,\Psi_2,\Psi_3,...,\Psi_q$ are independent and standard distribution variables, then the sum of their squares

\[Q=\sum^q_{k=1}{\Psi_k}\sim\chi^2\] 

Thus, $p$-value is able to be used for determine the probability of $(N-1)$th order permutation estimator. We use both $p$-values of Chi square distribution and Ratio comparison. 

\subsection{Application for any length of permutations}

We determine whether the sequence is biased or not by comparing the expected numbers with the observed numbers. This method allows for the expansion of the technique to accommodate permutations of any size through randomly selected experiments. The initial definition of the experiment was based on brute-force testing. However, by randomly selecting sample cases for experiments, we can represent the entire space with errors and measure the distance between the sample space and the original space.

\section{Fisher-Yates Shuffling and variants in an embedded system}\label{FYS}

The Fisher-Yates shuffling algorithm~\cite{FY}, is widely used as a hiding countermeasure. This is because it is straightforward for developers to implement, and it generates unbiased shuffling. The original method is described in Table~\ref{FS}. 

\begin{table}
    \centering
\setlength{\tabcolsep}{3pt}
\begin{tabular}{|p{300pt}|}
\hline
Input : An ordered Set $\{X_0,X_1,...,X_{t-1}\}$ where $X_i\in\mathbb{Z}_t$, \\
Output : An ordered Set $\{Y_0,Y_1,...,Y_{t-1}\}$ where $X_i\in\mathbb{Z}_t$\\
\hline
1. For $j$ from $t$ to $2$\\
~~~~1.1 $k$ is a random number between $1$ and $j$\\
~~~~1.2 Exchange $X_{k-1}$ and $X_{j-1}$\\
~~~~1.3 $Y_{j-1}\leftarrow X_{j-1}$\\
\hline
\end{tabular}
\caption{(Algorithm) Simple description of Fisher-Yates(FY) Shuffling}
\label{FS}
\end{table}

However, the algorithm carries potential risks that could result in a biased permutation, as explained below.

\begin{itemize}
\item \textbf{Incorrect Implementation} : Even with a correct algorithm in place, developers may still make mistakes in their code. Errors can easily occur when defining variables that determine the order or boundaries.
\item \textbf{Poor Random Source} : To generate a random permutation, random numbers derived from a 'good' random source are necessary.
\item \textbf{Application of Short Random Numbers} : When selecting a random integer within a given interval, there can be a bias in the selection. Short random numbers often cause this issue.
\end{itemize}

We will present two normal structures associated with random numbers, as well as two flawed structures, which serve as examples of incorrect implementation.

\subsection{Normal structures}
We mainly introduce two implementation for F-Y shuffling; Shuffling by reduction, and Shuffling by multiplication and division. 

\subsubsection{Shuffling by reduction}

The main point of difference for implementation of F-Y shuffling is choice of random number between $[1,j)$ the step 1.1 in Table~\ref{FS}. Because the value of variable $j$ is changed in each step, reduction operation is most widely used, see Algorithm in Table~\ref{knuth}. Reduction is normally high-redundancy operation.  

\begin{table}
    \centering
\setlength{\tabcolsep}{3pt}
\begin{tabular}{|p{300pt}|}
\hline
Input : An ordered Set $\{X_0,X_1,...,X_{t-1}\}$ where $X_i\in\mathbb{Z}_t$, and a source integer limitation $srt$\\
Output : An ordered Set $\{Y_0,Y_1,...,Y_{t-1}\}$ where $Y_i\in\mathbb{Z}_t$\\
\hline
1. For $j$ from $t$ to $2$\\
~~~~1.1 Generate a random integer $r\in\mathbb{Z}_{srt}$ \\
~~~1.2 $k=r$ mod $j$ (We expect $k$ is a random integer, between $0$ and $j-1$)\\
~~~~1.3 Exchange $X_{k}$ and $X_{j-1}$\\
~~~~1.4 $Y_{j-1}\leftarrow X_{j-1}$\\
\hline
\end{tabular}
\caption{(Algorithm) Shuffling by reduction}
\label{knuth}
\end{table}

\subsubsection{Shuffling by Multiplication and Division}
Here is another method for choosing a number in certain boundary as the Algorithm in Table~\ref{knuth_float}. We select a true number in \(r\leftarrow[0,1)\), and multiply the number and a integer $j$. Then, the integer part of $j\times r$ is in $0$ to $j-1$. If one can have trusted random number generator for floating point with $0$ to $1$, this algorithm generate a perfect permutation. Practically, we need to make a floating point in $0$ to $1$ with divide, such as generate a random number in $0$ to $u$, and divide $u$. As big as size of $u$, the number is generated uniformly. However, the redundancy of division operation is still problem.

\begin{table}
    \centering
\setlength{\tabcolsep}{3pt}
\begin{tabular}{|p{300pt}|}
\hline
Input : An ordered Set $\{X_0,X_1,...,X_{t-1}\}$ where $X_i\in\mathbb{Z}_t$, \\
Output : An ordered Set $\{Y_0,Y_1,...,Y_{t-1}\}$ where $Y_i\in\mathbb{Z}_t$\\
\hline
1. For $j$ from $t$ to $2$\\
~~~~1.1 Generate random number $r$, uniformly distributed between 0 to 1 (1 cannot be selected).\\
~~~1.2 $k\leftarrow \left \lfloor j\times r \right \rfloor+1$  (We expect $k$ is a random integer, between $1$ and $j$) \\
~~~~1.3 Exchange $X_{k-1}$ and $X_{j-1}$\\
~~~~1.4 $Y_{j-1}\leftarrow X_{j-1}$\\
\hline
\end{tabular}
\caption{(Algorithm) Shuffling by random floating point}
\label{knuth_float}
\end{table}

The Algorithm in Table~\ref{MD_shuffling} gives faster algorithm that is only switching multiplication and division. We can choose $rfn$ as $2^l$ then, division operation can be replaced by shift. It can be much faster than division and Shuffling by reduction.  

\begin{table}
    \centering
\setlength{\tabcolsep}{3pt}
\begin{tabular}{|p{300pt}|}
\hline
Input : An ordered Set $\{X_0,X_1,...,X_{t-1}\}$ where $X_i\in\mathbb{Z}_t$, and a refiner integer $rfn$\\
Output : An ordered Set $\{Y_0,Y_1,...,Y_{t-1}\}$ where $Y_i\in\mathbb{Z}_t$\\
\hline
1. For $j$ from $t$ to $2$\\
~~~~1.1 Generate a random integer $r\in\mathbb{Z}_{rfn}$ \\
~~~~1.1 $temp\leftarrow r\times j$\\
~~~~1.2 $k\leftarrow \left \lfloor temp/rfn \right \rfloor$  (We expect $k$ is a random integer, between $0$ and $j-1$, the operation $/$ can be replaced by shift where $rfn$ is power of 2) \\
~~~~1.3 Exchange $X_{k}$ and $X_{j-1}$\\
~~~~1.4 $Y_{j-1}\leftarrow X_{j-1}$\\
\hline
\end{tabular}
\caption{(Algorithm) Shuffling by Multiplication and Division}
\label{MD_shuffling}
\end{table}

\subsection{Bad structures}\label{bad_s}
In this section, we will introduce two well-known biased select functions for permutation: the Naive method and Sattolo's algorithm.

\subsubsection{NAIVE method}
The Naive method is implemented due to its superior speed compared to Fisher-Yates shuffling. The algorithm is delineated in Table \ref{NAIVE}. The primary difference between that and Fisher-Yates lies in the selection of numbers within a fixed boundary. If $t$ is $2^l$, this can be coded using bit-wise `And' operation. Proving this algorithm as biased or unbiased is challenging, so it's often used based on trust or the impression of randomness.

\begin{table}
    \centering
\setlength{\tabcolsep}{3pt}
\begin{tabular}{|p{300pt}|}
\hline
Input : An ordered Set $\{X_0,X_1,...,X_{t-1}\}$ where $X_i\in\mathbb{Z}_t$, \\
Output : An ordered Set $\{Y_0,Y_1,...,Y_{t-1}\}$ where $Y_i\in\mathbb{Z}_t$\\
\hline
1. For $j$ from $t$ to $2$\\
~~~~1.1 $k$ is a random number between $1$ and $t$\\
~~~~1.2 Exchange $X_{k-1}$ and $X_{j-1}$\\
~~~~1.3 $Y_{j-1}\leftarrow X_{j-1}$\\
\hline
\end{tabular}
\caption{(Algorithm) NAIVE method}
\label{NAIVE}
\end{table}

\subsubsection{Sattolo's Algorithm}
Sattolo's algorithm is a simple variant of the Fisher-Yates shuffling method that produces cyclic permutations of $(N-1)!$ instead of $N!$. This can be achieved by introducing a single index mistake during implementation, see Table~\ref{Satto}. The uniformity of this algorithm has also been proven for $(N-1)!$ cases~\cite{per_a2}. The advantage of this algorithm is that the result is a cyclic permutation. Although it is not necessary to use this algorithm for side-channel attack countermeasures, we introduce it as a case of failed implementation and present the experimental result in Section~\ref{Exp}.

\begin{table}
    \centering
\setlength{\tabcolsep}{3pt}
\begin{tabular}{|p{300pt}|}
\hline
Input : An ordered Set $\{X_0,X_1,...,X_{t-1}\}$ where $X_i\in\mathbb{Z}_t$, \\
Output : An ordered Set $\{Y_0,Y_1,...,Y_{t-1}\}$ where $Y_i\in\mathbb{Z}_t$\\
\hline
1. For $j$ from $t$ to $2$\\
~~~~1.1 $k$ is a random number between $1$ and $j-1$\\
~~~~1.2 Exchange $X_{k-1}$ and $X_{j-1}$\\
~~~~1.3 $Y_{j-1}\leftarrow X_{j-1}$\\
\hline
\end{tabular}
\caption{(Algorithm) Sattolo's algorithm}
\label{Satto}
\end{table}

\section{Security Analysis for real world environments}
\subsection{The experiment with Arbitrary sequences}
While we are able to verify perfect permutation using the $(N-1)$th order permutation criteria, we cannot manipulate each result to collect sequences under certain conditions. To verify arbitrarily generated permutations, we modify the $(N-1)$th order permutation to a "Rough" definition. Proposition~\ref{proposition_main} aids in creating this new definition to experiment with random sequences.

\begin{proposition}\label{proposition_main}
 A probability $P(A|B_1)=P(A|B_2)=...=P(A|B_n)=k$ and $(B_i\cap B_j)=\phi$ where $i\neq j$ then $P(A|\bigcup_{i=1}^{n}B_i)=k$.
\end{proposition}
\begin{proof}
The proof is in Appendix~\ref{proof_proposition}. 
\end{proof}

By the Proposition~\ref{proposition_main}, we can define Approximate $(N-1)$th order permutation below,

\begin{definition}
\textup{A select function $\gamma$ is said to possess the Approximate $(N-1)$th order permutation property if, for a randomly selected $t$ such that $\gamma(t) = f_m$, for all $m \in \mathbb{Z}_{N!}$ and distinct $x, y, a_{ij}, b_{ij} \in \mathbb{Z}_N$ where $x \neq a_{ij}$ and $y \neq b_{ij}$, the following condition holds:}

\[ P\left(f_m(x) = y \Big| \bigcup_{j=1}^{(N-2)!} \bigcap_{i=1}^{N-2} \{f_m(a_{ij}) = b_{ij}\}\right) = \frac{1}{2} \]

\textup{Additionally, it must satisfy the condition that for all $x, y \in \mathbb{Z}_N$:}

\[ P(f_m(x) = y) \neq 0 \]

\textup{where, for $k \neq l$, we have $a_{kj} \neq a_{lj}$ and $b_{kj} \neq b_{lj}$}

\end{definition}

For instance, in the $(N-1)$th order permutation definition for a 4-indices permutation, we have $P(f(2)=1|f(0)=2, f(1)=3)=1/2$. If we extend the Approximate $(N-1)$th order permutation, it becomes $P(f(2)=1|f(a_1)=b_1, f(a_2)=b_2)=1/2$, where $a\in A, b\in B, A={0,1}, B={2,3}$. According to Proposition~\ref{proposition_main}, if a select function $\gamma$ provides the $(N-1)$th order permutation, it also provides the Approximate $(N-1)$th order permutation. However, if $\gamma$ has the Approximate $(N-1)$th order permutation, it doesn't necessarily guarantee the $(N-1)$th order permutation.

Therefore, we can infer that if the select function does not have the Approximate $(N-1)$th order permutation, it lacks perfect random security.

\subsection{Error and limitation}

The error probability between the Approximate $(N-1)$th order permutation and perfect random permutation can be computed by evaluating the probability that makes the reverse statement of Proposition~\ref{proposition_main} true. Specifically, if the probability $P(A|\bigcup_{i=1}^{M}B_i)=k$ and $(B_i\cap B_j)=\phi$ where $i\neq j$, then $P(A|B_1)=P(A|B_2)=...=P(A|B_M)=k$.

Setting $P(A|B_1)=p_1, P(A|B_1)=p_2, ..., P(A|B_M)=p_M$, \newline
we have $P(A|\bigcup_{i=1}^{M}B_i)=(p_1P(B_1)+p_2P(B_2)+...+p_MP(B_M))/(P(B_1)+P(B_2)+...+P(B_M))=k$. We can express this equation as $\sum_{i=1}^{M}p_iP(B_i)=k\times\sum_{i=1}^{M}P(B_i)$. If $\gamma$ has $(N-1)$th order permutation case $p_1=p_2=...=p_M$ equals to $k$. Hence, the error probability $\mathbf{e}$ is $P(\sum_{i=1}^{M}p_iP(B_i)=k\times\sum_{i=1}^{M}P(B_i))$ and there must exist at least one $p_i\neq k$.

Intuitively, the value of the error probability is extremely small due to the large number of variables involved. For $(N-1)$th order permutation, $M=(N-2)!$, and the computation involving the combination of each variable $\sum_{i=1}^{M}p_iP(B_i)$ must equal $k\times\sum_{i=1}^{M}P(B_i)$.

If the number of variables in the equation decreases, the error probability would increase. Therefore, to minimize error values, we must verify whether the algorithm (or select function $\gamma$) generates all possible permutations, regardless of potential bias.

If one intentionally generates certain sequences, the complexity distance between sequences verified by the Approximate $(N-1)$th order permutation and perfect random security is substantial. However, unintentionally biased sequences, generated for randomness, are well-verified by the Approximate estimator. As a toy example of a permutation which is an Approximate $(N-1)$th order permutation but not $(N-1)$th order, refer to Figure~\ref{constrained}. Assuming that 4! cases are uniformly distributed initially, each case has an equal probability of occurrence $1/24$. We set the first case (1,2,3,4) to be biased with $\alpha_1\neq 1/24$, while the others remain uniformly distributed. Then, given that the permutation is an Approximate $(N-1)$th order, $\alpha_1+\alpha_2=\alpha_7+\alpha_8$, and $\alpha_1+\alpha_2+\alpha_7+\alpha_8=4/24$. Thus, $\alpha_1+\alpha_2$ is $2/24$ and $\alpha_2$ is also determined. $\alpha_1$ influences 6 values with a single transposition from (1,2,3,4), for example (1,2,4,3) that is transposed by (3,4). In this manner, the 6 values determined by $\alpha_1$ can affect the other 6 values via related equations such as $\alpha_1+\alpha_3=\alpha_{22}+\alpha_{24}$ in Figure~\ref{constrained}. As a result, the probabilities of each case depend on other cases. Despite the many variables, the biased shape appears quite systematic. The final solution of this example for each probability is classified into only two values $\alpha_1$ and $\alpha_2$. In this case, the set of even permutation ($A_4$) is an exceptional case of Approximate $(N-1)$th order permutation but not perfect random permutation. However, it is extremely hard that a set of permutations which is generated by an arbitrary algorithm becomes $A_n$.  

There are two primary non-uniform permutation models: the biased model and the constrained model~\cite{Phd_T}. The biased model refers to an algorithm that aims for uniformity but manifests a biased form. The constrained bias model, on the other hand, shows a biased pattern in a specific shape. Therefore, the case of an Approximate $(N-1)$th order permutation, which is not $(N-1)$th order, typically represents the constrained model. Conversely, the Fisher-Yates shuffling algorithm generates a perfect random permutation, so we can anticipate verification with the approximate estimator.

\begin{figure*}

\includegraphics[width=13cm]{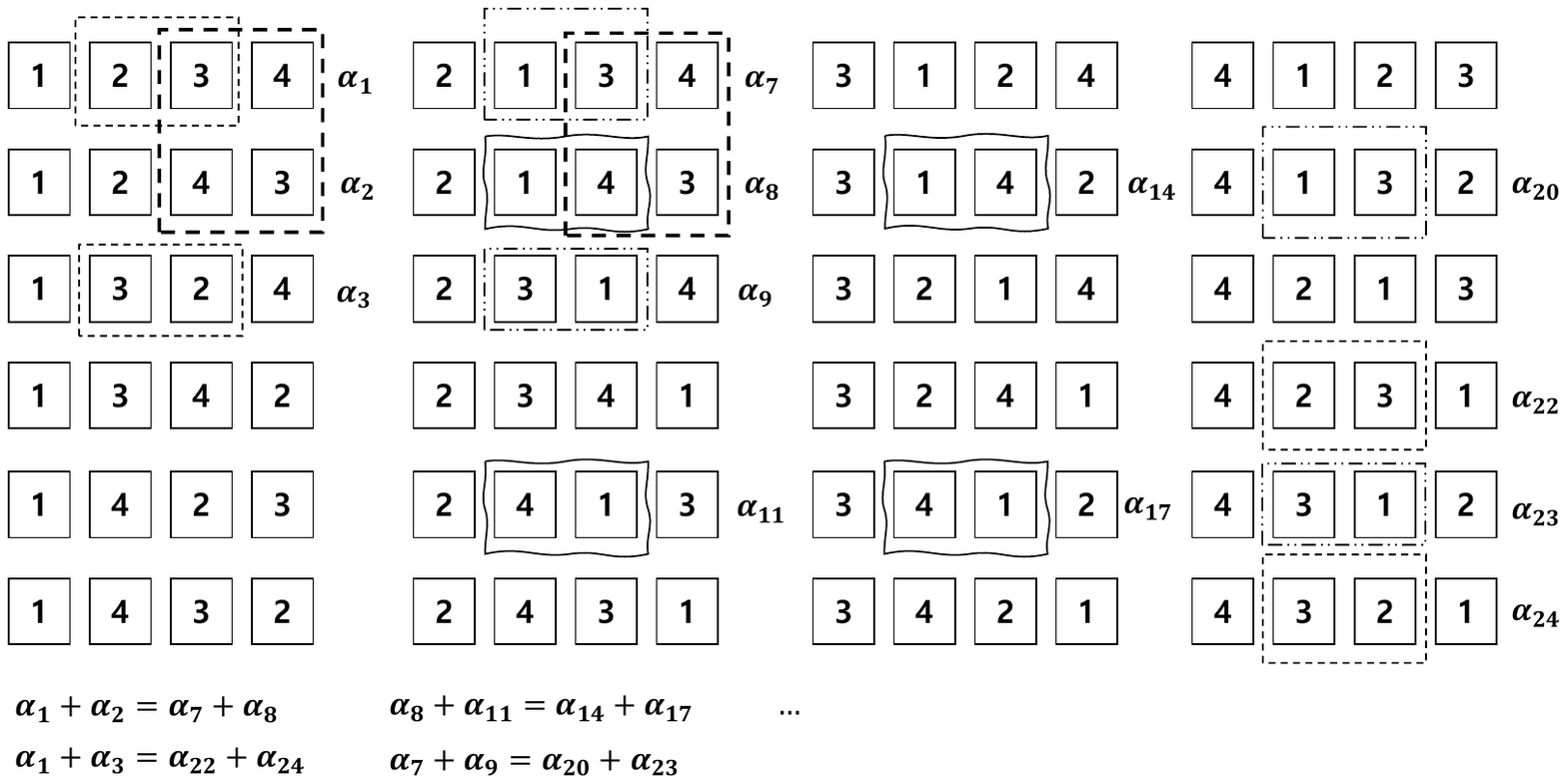}
\centering
\caption{A progress to become for an Constrained model of Approximate $(N-1)$th order permutation; from $(N-1)$th order permutation where $N=4$.}
\label{constrained}
\end{figure*}

\subsection{Time and Space for analysis}

\begin{table}
\setlength{\tabcolsep}{3pt}
\begin{tabular}{|p{330pt}|}
\hline
Input : $RN$(Random string, given length and properties), select function $\gamma$ \\
\hline
// Collecting Phase, Collect $M=(N!\times T)$ permutations by $\gamma$\\
1. $i$ from $0$ to $M-1$\\
~~~~~~1.1 $s_i = \gamma(RN_i)$~~// $RN_i$ is $i-th$ segment in $RN$.\\
// 2. Analysis Phase \\
2. $i$ from $0$ to $M-1$\\
~~~~~~2.1 Get numbers for each case from the permutation $s_i$\\
// There are $N!$ cases, update the count. \\
// Analysis Phase - Get Statistics and final result\\
3. $i$ from $0$ to $N!$\\
~~~~~~3.1 Get $\rho_i$ // normalized values of each case\\
~~~~~~3.2 SUM = SUM + $\rho_i$
// Analysis Phase 2 - Getting final result \\
4. return $p$-values of Chi-Square($\chi^2$) from SUM \\
\hline
\end{tabular}
\centering
\caption{Pseudo code of brute-force analysis}
\label{Brute_al}
\end{table}

\begin{table}
\setlength{\tabcolsep}{3pt}
\begin{tabular}{|p{330pt}|}
\hline
Input : $RN$(Random string, given length and properties), select function $\gamma$ \\
\hline
// Collecting Phase, Collect $M=T\times\binom{N}{N-2}\times \binom{N}{N-2}$ permutations by $\gamma$\\
1. $i$ from $0$ to $M-1$\\
~~~~~~1.1 $s_i = \gamma(RN_i)$~~// $RN_i$ is $i-th$ segment in $RN$.\\
// 2. Classify Phase \\
2. $i$ from $0$ to $M-1$\\
~~~~~~2.1 Getting numbers for each location from the permutation $s_i$\\
// There are $\binom{N}{N-2}$ locations, update the count. \\
// Analysis Phase 1 - Getting Statistics\\
3. $i$ from $0$ to $\binom{N}{N-2}\times \binom{N}{N-2}$\\
~~~~~~3.1 Get Function values($p_i$) by Probability density of each binomial distribution. \\
// Analysis Phase 2 - Getting final result \\
4. return ratio values(for ratio comparison) and $p$-values by Chi-Square($\chi^2$) distribution. \\
\hline
\end{tabular}
\centering
\caption{Pseudo code of Approximate $(N-1)$th order analysis}
\label{PCA}
\end{table}

\begin{table}
\caption{Time and Space for the permutation analysis}
\setlength{\tabcolsep}{3pt}
\begin{tabular}{|p{120pt}|p{80pt}|p{120pt}|}
\hline
   &   Brute force analysis & Approximate $(N-1)$th order  \\
\hline
Time for Collecting Samples  & $\mathcal{O}(T\times N!)$ & $\mathcal{O}(T\times N(N-1)/2)$  \\
Space for Collecting Samples & $\mathcal{O}(T\times N!)$ & $\mathcal{O}(T\times N(N-1)/2)$  \\
Space for Analysis&$\mathcal{O}(T\times N!)$ &$\mathcal{O}(T\times N^2(N-1)^2/4)$  \\
Time for Analysis&$\mathcal{O}(T\times N!)$ &$\mathcal{O}(T\times N^2(N-1)^2/4)$ \\
\hline
\end{tabular}
\centering
\label{TiSp}
\end{table}

\begin{figure}
\includegraphics[width=12cm]{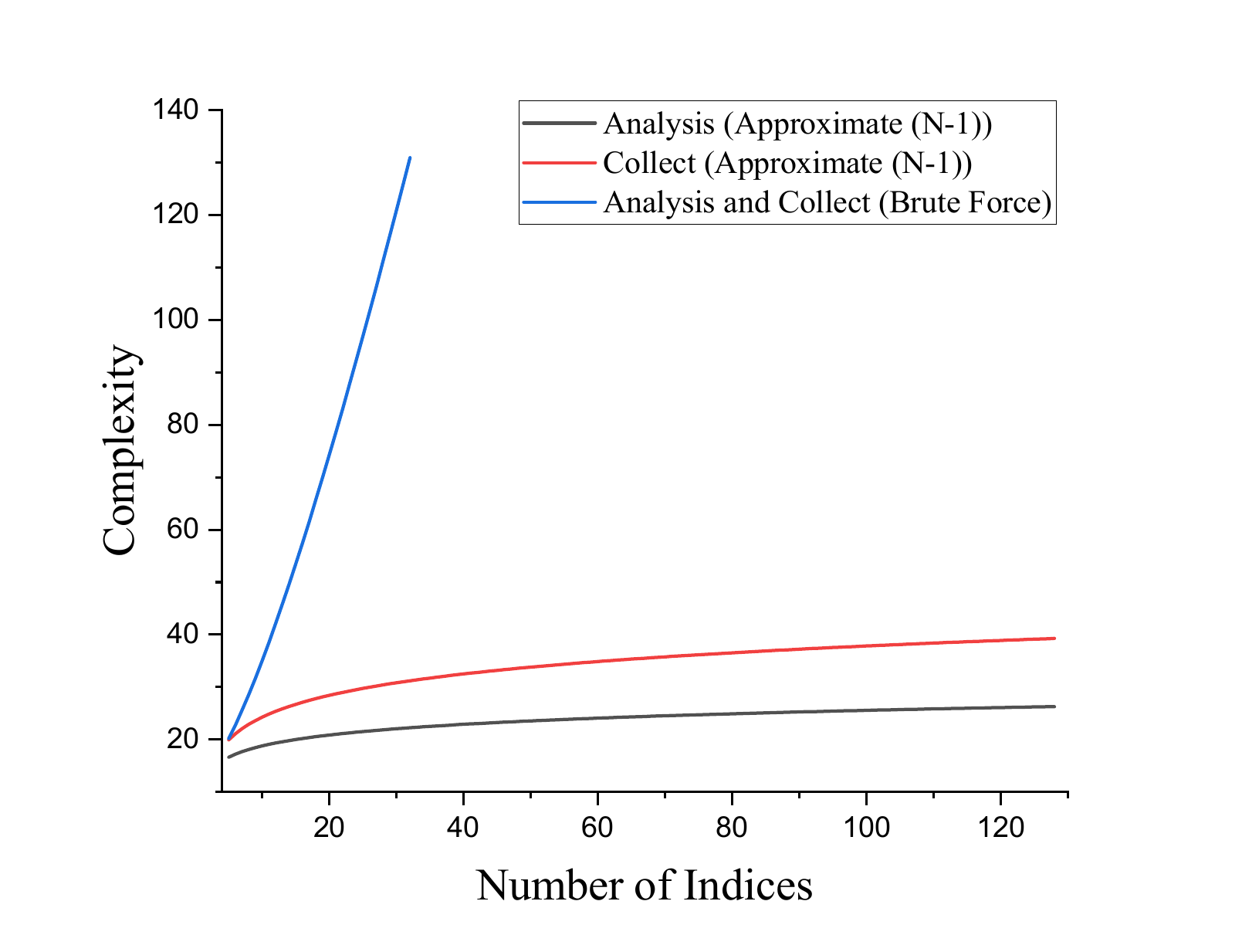}
\centering
\caption{Increasing of Time and Space for Approximate $(N-1)$th and Brute-Force Test. (The time and space complexity for Brute-force)}
\label{ITS_PIC}
\end{figure}

The time and space complexity for the analysis method are outlined in Table~\ref{TiSp}. $T$ is the number of samples for each test. We compare this with brute force sampling and the $\chi^2$ test for uniform distribution with collected permutations, as seen in Algorithm Table~\ref{Brute_al}. In brute-force sampling, $T\times N!$ samples are required for one significant test. Practically, we use the Approximate $(N-1)$th order estimate in Table~\ref{PCA}, which requires $\mathcal{O}(T\times (N(N-1))/2)$ time for analysis and sampling. This might seem like it's calculated based on the number of two selections $N(N-1)/2$ of inputs and outputs. However, each permutation is reusable for each input. If $N$ is 32, then the time and space complexity is $496\times T$. In our experimental result, we generate $10^7$ samples, so $T$ equals $10^7/496=21,161$, which represents the number of samples for each test. However, the complexity of analysis is $\mathcal{O}(T\times N^2(N-1)^2/4)$, because we want to obtain all results about possible inputs and outputs. The collection phase and the classification phase can be merged into a single phase, so no extra space for collected permutations is required.

Figure~\ref{ITS_PIC} illustrates the increase in space and time for the analysis on a $\log_2$ scale. The y-axis represents the bit-wise complexity of the permutation, which shows how the speed of time and space for the brute-force approach increases in relation to $N$. It is impossible to perform brute force analysis when $N$ is exceedingly large. Also, the approximate estimator is practically unfeasible when $N>50$, although this depends on the environment.

\begin{table}
\caption{Actual Time and Space for the permutation analysis}
\setlength{\tabcolsep}{3pt}
\begin{tabular}{|p{100pt}|p{50pt}|p{70pt}|p{70pt}|}
\hline 
                        -                     & $N$                                         & Brute force analysis   & Approximate $(N-1)$th order \\ 
\hline                                                                                                                          
Actual Time                                  &$N=8$                                        & 241 Sec.               & 0.1 Sec.                    \\ 
for collecting Samples                       &$N=32$                                       & $5^{27}$ Years         & 3 Sec.                      \\ 
                                             &$N=128$                                      & $7.3^{205}$ Years      & 48 Sec.                     \\ 
                                                                                                                                
\hline                                                                                                                          
Actual Space                              &$N=8$                                        & 24GB                  & 2187KB                       \\ 
for collecting Samples                    &$N=32$                                       & $2.4^{30}$TB          & 151MB                        \\ 
(Can be omitted)                          &$N=128$                                      & $5.7^{211}$TB         & 9.6GB                        \\ 

\hline                                                                                                                          
Actual Space for           &$N=8$                                        & 252KB             & 3.1KB                  \\ 
Analysis                   &$N=32$                                       & $9.5^{25}$TB      & 961KB                  \\ 
                           &$N=128$                                      & $1.4^{207}$TB     & 252MB                  \\

\hline                                                                                                                          
Actual Time for            &$N=8$                                        & 40 Sec.               & 15 Sec.                      \\
Analysis                   &$N=32$                                       & $8.3^{24}$ Years      & 1.3 Hours                    \\ 
                           &$N=128$                                      & $1.2^{205}$ Years     & 15.2 Days                    \\

\hline
\end{tabular}
\centering
\label{ATiSp}
\end{table}

Table~\ref{ATiSp} displays the actual time and space required for the analysis of each method, based on an Intel(R) Core(TM) i5-8500 CPU @ 3.00GHz. The results can vary with different implementation or improvements in the permutation algorithms. As a result, brute force analysis is practically unfeasible in the real world. While our method also requires a substantial amount of time and sizable space when $N$ is 128, it can still analyze a reduced number of cases.

\subsection{Reduced cases for Large scale permutations}

Another advantage of our estimator is that even if only some of the generated permutations are tested, the results similar to those of a complete test can be obtained. This characteristic can be utilized effectively to verify large permutations. For instance, if we verify the permutation of $N=128$, the analysis time is $(128\times 127)^2/4=66,064,384$. However, as all test cases independently follow the standard normal distribution, similar trend results can be obtained by experimenting with only a subset of them. Therefore, even if $N=128$, it is feasible to select and test an appropriately-sized set of input and output from the $66,064,384$ test cases. A comparison of the overall test results with the reduced test case will be presented in the experimental results.

\section{Experimental Result}\label{Exp}

\subsection{Experimental Result - F-Y shuffling variants}\label{compare_FY}
We choose $N=32$ as the input message size for a single trace attack on CRYSTALS Kyber. We conduct experiments based on the Approximate $(N-1)$th order permutation with the following steps:

\begin{itemize}
    \item Collect a large number of arbitrary random sequences from the target algorithm and random inputs. In our experiments, we generate $10^7$ permutations.
    \item Classify the sequences based on $\binom{N}{N-2}$ input conditions and $\binom{N}{N-2}$ output conditions, equivalent to $\binom{N}{2}\times \binom{N}{2}$ cases. In our experiments where $N=32$, the total number of cases for verification is $246016$.
    \item Determine the quality of outputs using the security determination method explained in section~\ref{SD}.
    \begin{itemize}    
        \item Classify each case by the probability $\alpha$, as outlined in Table~\ref{xaxis}. The number of cases where $P(x\geq X)>\alpha$, Here, $X\sim\mathcal{N}(\mu,\sigma^2)$ and $\alpha$ is the probability criteria as shown in Table~\ref{xaxis}.
        \item Compare the expected number$(EN)$ of cases with the observed numbers $(ON)$. The comparison result is expressed as the log scale of the ratio between observed numbers$(ON)$ and expected numbers$(EN)$, denoted as $log_2(ON/EN)$.
    \end{itemize}
\end{itemize}

\begin{table}
\setlength{\tabcolsep}{3pt}
\begin{tabular}{|p{90pt}|p{90pt}|p{90pt}|}
\hline
Index(x-axis for Figure~\ref{rng_reduction},\ref{rng_fast},\ref{compare_pic}, and \ref{sample_result})& $\alpha$ & Expected Number($EN$) of cases, $N=32$, $\alpha\times 246016$
\\
\hline
1&0.55&221415\\
2&0.60&196813\\
3&0.70&147610\\
4&0.80&98406\\
5&0.90&49203\\
6&0.95&24601\\
7&0.99 (=$1-10^{-2}$)&4920.34\\
8& $1-10^{-3}$ &492.034\\
9& $1-10^{-4}$ &49.20\\
10& $1-10^{-5}$ &4.92034\\
11& $1-10^{-6}$ &0.49203\\
12& $1-10^{-7}$&0.04920\\
13& $1-10^{-8}$&0.004920\\
14& $1-10^{-9}$&0.000492\\
15& $1-10^{-10}$&0.000049\\
16& $1-10^{-11}$&0.000005\\
\hline
\end{tabular}
\centering
\caption{The value of confidence level(The x-axis of each Figures)}
\label{xaxis}
\end{table}

As discussed in Section~\ref{FYS}, there are three types of error cases in permutation generation. To test these possibilities, we conducted experiments using the following classifications:

Firstly, to evaluate the effect of the random source on the implementation, we utilized two different random functions. As an example of a non-random source, we used the built-in rand() function from Microsoft Visual C's standard library~\cite{MSDN}. This function uses a technique known as a linear congruential generator with a period of $2^{16}$, which is an example of a random source unsuitable for cryptographic purposes. For a well-implemented random source, we simulated a random number generator using AES. We conducted a loop with a fixed key, using the output value from the initial input value as the next input value.

Secondly, we investigated the length of the random number used in the algorithms. We could experiment with bit selection from 5 bits for shuffling 32 positions. However, a bias would obviously appear for 5 bits, and since implementers gain little performance benefit from using less than 1 byte, we applied a minimum of 8 bits and a maximum of 15 bits.

Thirdly, we examined the impact of incorrect implementation. For this, we implemented and tested the NAIVE and Sattolo methods, as described in Section~\ref{bad_s}. To isolate the results due to faulty implementation, we used 16 bits and AES (which is 1 bit larger than the maximum bit).

\begin{figure}[H]
\centering
\begin{subfigure}{0.45\linewidth}
\caption{}
\label{norng_reduction}
\includegraphics[width=\linewidth]{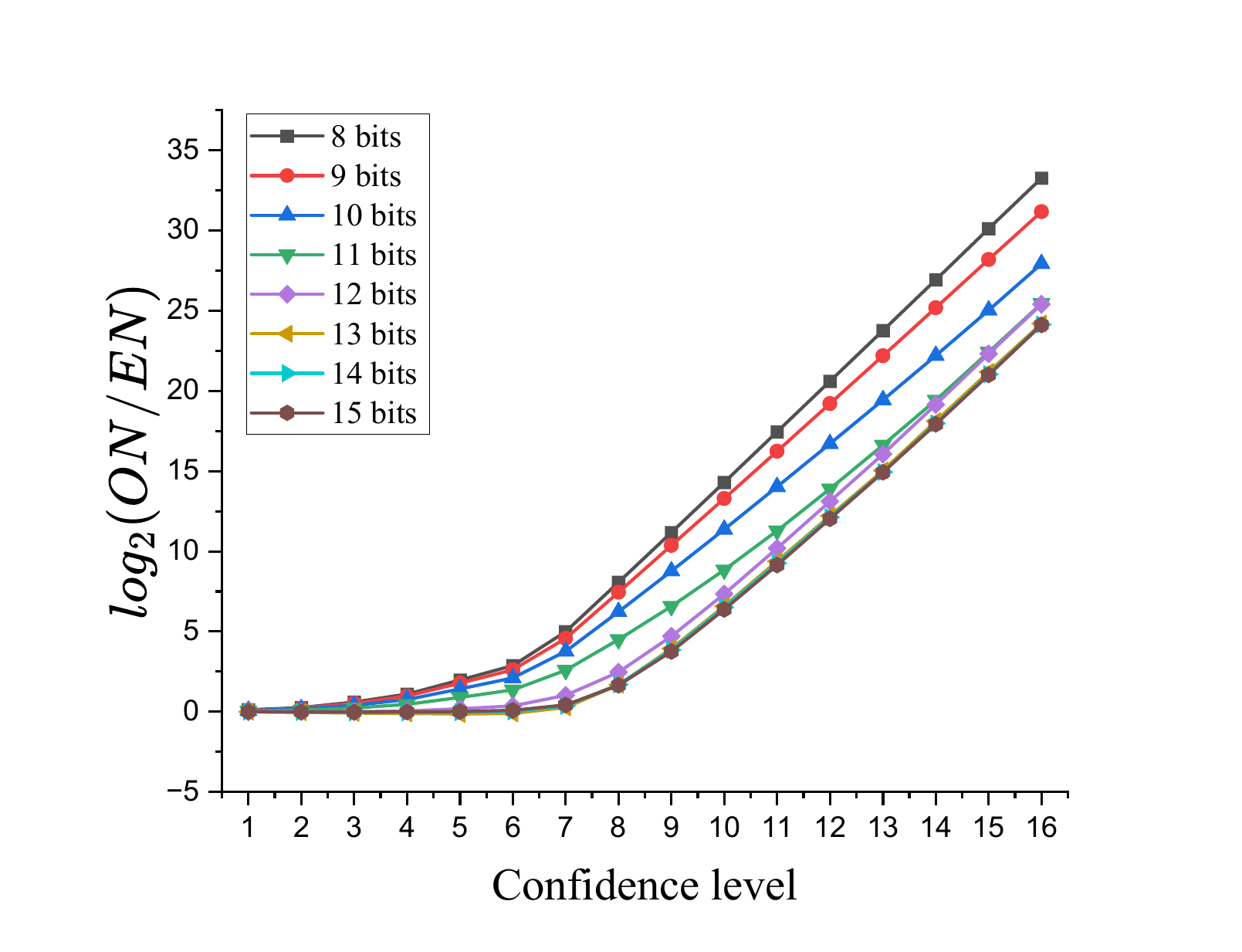}
\end{subfigure}
\hfill
\begin{subfigure}{0.45\linewidth}
\caption{}
\label{rng_reduction}
\includegraphics[width=\linewidth]{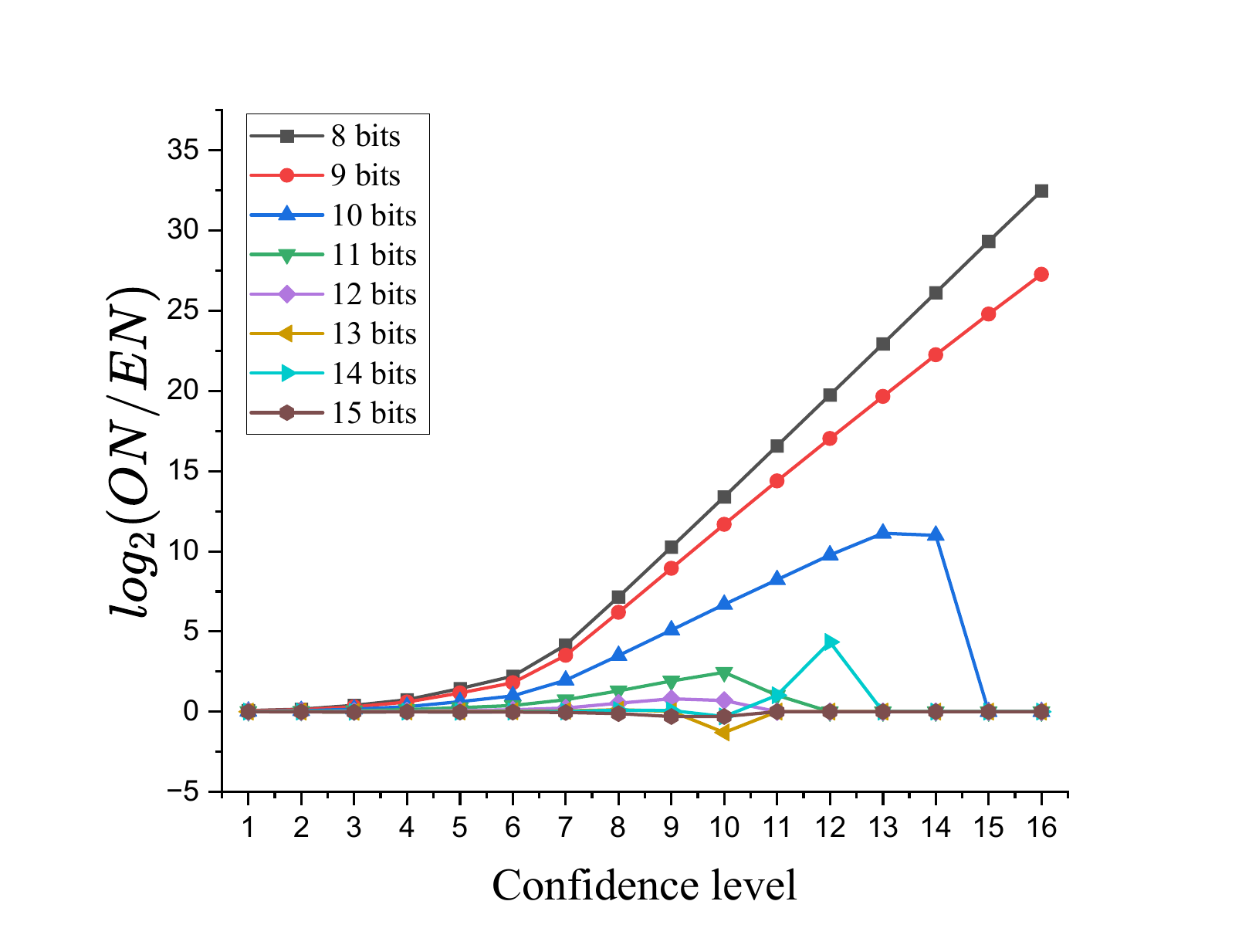}
\centering
\end{subfigure}
\centering

\caption{(a) F-Y shuffling with Standard library RNG and Reduction and (b) F-Y shuffling with AES RNG and Reduction}
\end{figure}

\begin{figure}[H]
\centering
\begin{subfigure}{0.45\linewidth}
\caption{}
\label{norng_fast}
\includegraphics[width=\linewidth]{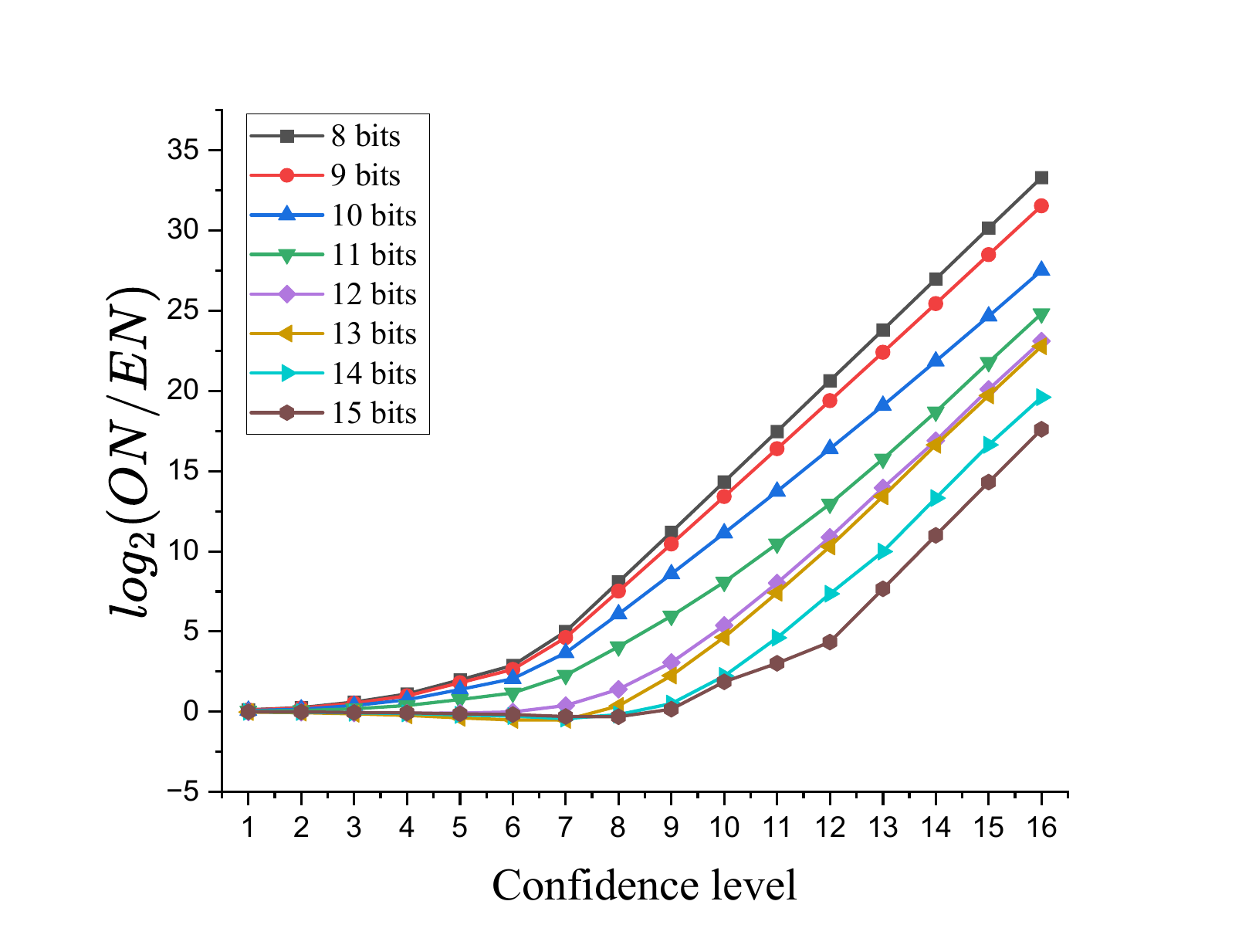}
\end{subfigure}
\hfill
\begin{subfigure}{0.45\linewidth}
\caption{}
\label{rng_fast}
\includegraphics[width=\linewidth]{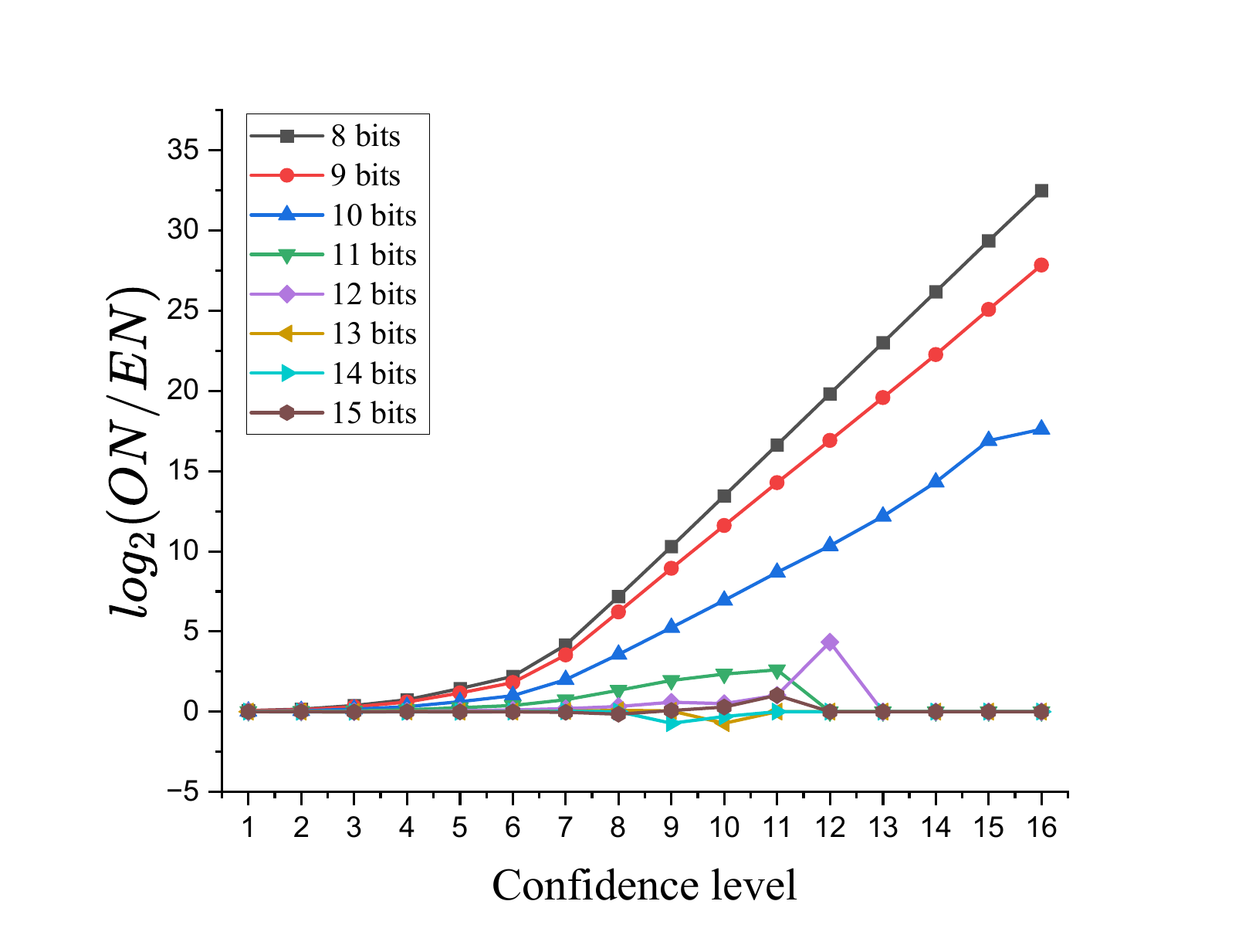}
\centering
\end{subfigure}
\centering

\caption{(a) F-Y shuffling with Standard library RNG and multiplication-division and (b) F-Y shuffling with AES RNG and multiplication-division}

\end{figure}

\subsubsection{Analysis of the each graph(Ratio Comparison)}

The x-axis represents the value of the probability density level, as shown in Table~\ref{xaxis}. In the result of $246016$ cases, since each case follows a normal distribution, we can compute the expected number of cases for $\alpha$ as $\alpha \times 246016$. The y-axis is a log scale of the ratio between observed numbers $(ON)$ and expected numbers $(EN)$, denoted as $Log_2(ON/EN)$. 

If $\gamma$ exhibits an Approximate $(N-1)$th order permutation, the expected number $(EN)$ will be very close to the observed number $(ON)$, implying that $(ON)/(EN)\approx 1$. Therefore, if $\gamma$ has perfect random security, the values on the y-axis for each line should be close to 0. Conversely, if the values on the y-axis for each line are high, this indicates bias. The value on the y-axis reveals the degree to which the permutation deviates from perfect randomness.

Figure~\ref{norng_reduction} presents the $(N-2)$th order permutation evaluation from 8 bits to 15 bits, with different colors and symbols representing the standard library rand() source and Fisher-Yates shuffling by reduction, as described in the algorithm in Table~\ref{knuth}. Despite the y-axis being expressed on a log scale, the results show slight improvement from 8 bits to 15 bits. A value of $10 (=35-25)$ represents a $2^{10}$ difference.

Figure~\ref{norng_fast} displays the results with an AES-based RNG and shuffling by reduction. The results suggest that 11 bits of random number are sufficient to obtain random permutations if the random source is reliable.

\begin{figure}

\includegraphics[width=12cm]{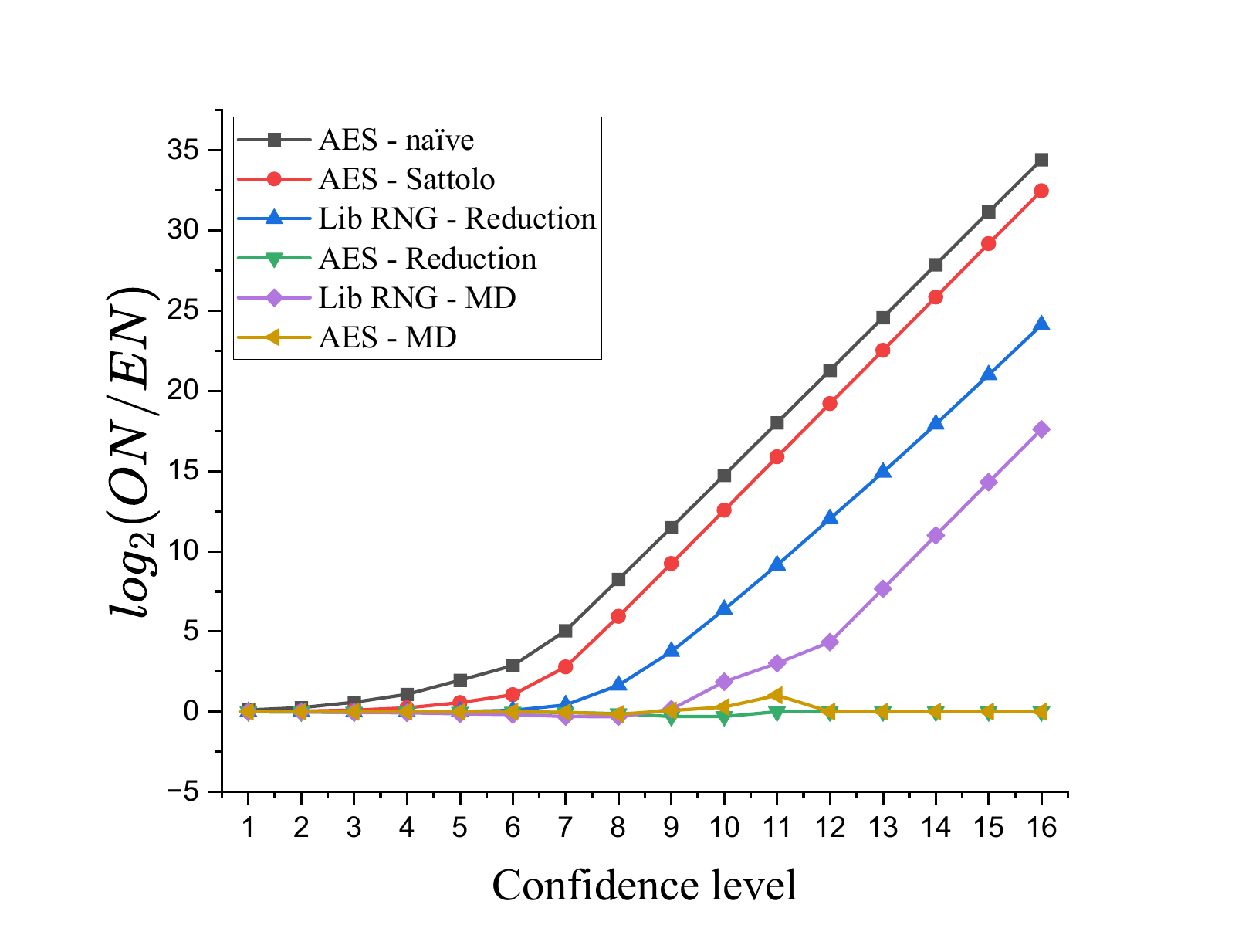}
\centering

\caption{NAIVE 16 bits, Ssatolo 16 bits and Fisher-Yates variants 15bits; y-axis is ratio of (ON/EN) and x-axis is the confidence level given by Table~\ref{xaxis}. Having a high value implies having a high bias (a bad outcome).}
\label{compare_pic}
\end{figure}

Figure~\ref{norng_fast} and Figure~\ref{rng_fast} display results using the standard library rand() and an AES-based RNG, respectively, with the shuffling implementation employing the Multiplication and Division method, as outlined in Algorithm Table~\ref{MD_shuffling}. The results are quite similar to those obtained through shuffling by reduction. Interestingly, the Multiplication and Division method yields better results than shuffling by reduction when using a subpar random source.

This can be directly compared in Figure~\ref{compare_pic}, which depicts two examples of poor structures, Ssatolo and NAIVE, with a 16-bit RNG source. Despite utilizing an AES-based RNG and 16-bit randomness, these results are worse than the 15-bit outcomes of standard F-Y shuffling, regardless of the implementation. The comparison between the MD shuffling and reduction methods in Figure~\ref{compare_pic}, at x-axis 16 (the case of $\alpha=1-(1/10^{11})$), shows results of approximately 17 and 24. We can therefore conclude that MD shuffling is about $2^7$ times better with the Approximate $(N-1)$th order permutation criterion.

In terms of the implementation of F-Y shuffling, the MD algorithm is not only significantly faster (when $N=2^r$) than reduction, but also generates greater randomness. Thus, we strongly recommend using the MD algorithm for F-Y shuffling, particularly when dealing with an unreliable random source.

\subsubsection{Analysis by $p$-value of Chi Square($\chi^2$)}

Table~\ref{reault_chi_all} presents the results of the Chi-Square ($\chi^2$) test from 8 bits to 15 bits with F-Y reduction and the Multiplication-Division method. Each number can be interpreted as the error probability when asserting that the distribution is biased. In other words, a value closer to 1 indicates a more uniform distribution, while a value closer to 0 points to a biased distribution.

Curiously, in some cases, the $\chi^2$ test identifies distributions as uniform even when the graph analysis suggests bias. Both analyses occur within the Approximate $(N-1)$th order estimator, so users can verify using both methods. The $\chi^2$ test's blind spot is that it may not detect a single, extremely anomalous case, the significance of which depends on how statistical test results are interpreted. Conversely, the strength of graph analysis lies in its ability to capture the occurrence of a single normal case and analyze the probability. Therefore, these two methods can be used in a complementary manner.

\begin{table}
\centering
\begin{tabular}{|p{30pt}|p{50pt}|p{50pt}|p{50pt}||p{50pt}|}
\hline
&  \multicolumn{2}{|c|}{Reduction} & \multicolumn{2}{|c|}{M-D} \\
\hline
RNG Bits & rand() & AES & rand() & AES  \\
\hline

8 &    \underline{$<10^{-5}$} &   \underline{$<10^{-5}$} &   \underline{$<10^{-5}$} &   \underline{$<10^{-5}$}   \\ 
9 &   \underline{$<10^{-5}$} &   \underline{$<10^{-5}$} &   \underline{$<10^{-5}$} &   \underline{$<10^{-5}$}   \\ 
10 &   \underline{$<10^{-5}$} &   \underline{$<10^{-5}$} &   \underline{$<10^{-5}$} &   \underline{$<10^{-5}$}   \\ 
11 &   \underline{$<10^{-5}$} &   \underline{$<10^{-5}$} &   \underline{$<10^{-5}$} &   \underline{$<10^{-5}$}   \\ 
12 &   \underline{$<10^{-5}$} &   \underline{$<10^{-5}$} &   \underline{$<10^{-5}$} &   \underline{$<10^{-5}$}   \\ 
13 &   \underline{$<10^{-5}$} &   \underline{$<10^{-5}$} &   1.00000 &                            \underline{0.00062}   \\ 
14 &   \underline{$<10^{-5}$} &   \underline{0.03863} &                   1.00000 &                                     0.27722   \\ 
15 &   \underline{$<10^{-5}$} &   0.99856 &                            1.00000 &                                     0.58224 \\ 

\hline
\end{tabular}
\caption {$1-$($p$-value) of Chi-square($\chi^2$) test by different random sources, two F-Y implementations, and input random bits of Approximate $(N-1)$th order estimation. (\underline{underline} number shows that $1-$($p$-value) is under $0.05$)}
\label{reault_chi_all}
\end{table}

\subsection{Comparison between brute-force and Approximate $(N-1)$th order analysis}

\begin{table*}

\centering
\begin{tabular}{|p{40pt}|p{40pt}|p{40pt}|p{40pt}||p{40pt}|p{40pt}|p{40pt}|}
\hline
&  \multicolumn{3}{|c|}{Brute-force test} & \multicolumn{3}{|c|}{Approximate $(N-1)$th estimation} \\
\hline
RNG Bits & $6$ & $7$ & $8$ & $6$ & $7$ & $8$ \\
\hline
24&  0.97071&  0.78157&  0.99129&  0.58029&  0.99470&  0.59607\\  
23&  0.60195&  0.85327&  0.22004&  0.76510&  0.96337&  0.93582\\  
22&  0.89910&  0.24739&  0.86581&  0.97857&  0.58073&  0.54206\\  
21&  0.18829&  0.40800&  0.52338&  0.05288&  \underline{0.00082}&  \underline{0.03122}\\  
20&  0.35210&  0.52689&  0.39091&  0.62381&  0.78894&  0.16322\\  
19&  0.96127&  0.17897&  0.65415&  0.97312&  0.75729&  \underline{0.01795}\\  
18&  0.24515&  0.86401&  0.86514&  \underline{0.00010}&  0.99973&  0.34465\\  
17&  \underline{0.03038}&  0.36051&  0.89479&  0.37356&  0.11754&  0.68689\\  
16&  0.45803&  0.37862&  0.82553&  0.83978&  0.24456&  0.44144\\  
15&  0.14034&  0.53689&  0.63005&  \underline{0.02683}&  0.19706&  0.40542\\  
14&  \underline{0.04169}&  0.86262&  0.05209&  0.35372&  0.99926&  \underline{0.00194}\\  
13&  \underline{0.01777}&  \underline{0.00799}&  0.45568&  \underline{$<10^{-5}$}&  \underline{0.00001}&  \underline{0.00414}\\  
12&  \underline{0.01102}&  0.21629&  0.31543&  \underline{$<10^{-5}$} & \underline{$<10^{-5}$}&  \underline{$<10^{-5}$}\\  
11&  \underline{$<10^{-5}$}&  \underline{$<10^{-5}$}&  \underline{0.00018}&  \underline{$<10^{-5}$}& \underline{$<10^{-5}$}& \underline{$<10^{-5}$}\\  
10&  \underline{$<10^{-5}$}&  \underline{$<10^{-5}$}&  \underline{$<10^{-5}$}&  \underline{$<10^{-5}$}& \underline{$<10^{-5}$}& \underline{$<10^{-5}$}\\  
9&   \underline{$<10^{-5}$}&  \underline{$<10^{-5}$}&  \underline{$<10^{-5}$}&  \underline{$<10^{-5}$}& \underline{$<10^{-5}$}& \underline{$<10^{-5}$}\\  
8&   \underline{$<10^{-5}$}&  \underline{$<10^{-5}$}&  \underline{$<10^{-5}$}&  \underline{$<10^{-5}$}& \underline{$<10^{-5}$}& \underline{$<10^{-5}$}\\  
7&   \underline{$<10^{-5}$}&  \underline{$<10^{-5}$}&  \underline{$<10^{-5}$}&  \underline{$<10^{-5}$}& \underline{$<10^{-5}$}& \underline{$<10^{-5}$}\\  
6&   \underline{$<10^{-5}$}&  \underline{$<10^{-5}$}&  \underline{$<10^{-5}$}&  \underline{$<10^{-5}$}& \underline{$<10^{-5}$}& \underline{$<10^{-5}$}\\  
5&   \underline{$<10^{-5}$}&  \underline{$<10^{-5}$}&  \underline{$<10^{-5}$}&  \underline{$<10^{-5}$}& \underline{$<10^{-5}$}& \underline{$<10^{-5}$}\\  
4&   \underline{$<10^{-5}$}&  \underline{$<10^{-5}$}&  \underline{$<10^{-5}$}&  \underline{$<10^{-5}$}& \underline{$<10^{-5}$}& \underline{$<10^{-5}$}\\
\hline
\end{tabular}
\caption{$1-$($p$-value) of Chi-square($\chi^2$) test by the number of inputs are $6$, $7$, and $8$ by Brute-force test and Approximate $(N-1)$th order estimation. (\underline{underline} number shows that $1-$($p$-value) is under $0.05$)}
\label{compare_perfect}
\end{table*}

To verify the reliability of the Routh $(N-1)$th order permutation estimator, we compared its results with a brute-force test of all permutations with a small number of inputs and outputs; specifically, $6$, $7$, and $8$. We used $10^8$ permutations with the Multiplication and Division method with an AES-based RNG, varying the input bits from $24$ to $4$. Given that we used $10^8$ permutations, there are $6!$, $7!$, and $8!$ permutations for each respective case.

Table~\ref{compare_perfect} displays both the total survey results and the results when using the Approximate $(N-1)$th order estimator. To accurately interpret these results, recall Proposition~\ref{proposition_main}: the Approximate test is highly effective for detecting bias. However, it may produce errors when concluding that a distribution is uniform. Therefore, it's important to watch for potential errors in the Brute-force test, where it may conclude that a bias is present, and the Approximate test, where it may identify the distribution as uniform. These discrepancies are rare and only occur when the bias is analyzed at $95\% (1-p = 0.95)$. We observed one such case at 13 bits($N=6$) and 17 bits($N=6$), respectively, though these occurrences may vary depending on the $p$-value. On the other hand, additional bias was confirmed in the Approximate test, which actually demonstrated better results in the Approximate test. Ultimately, we concluded that the approximate $(N-1)$th order estimator is as reliable as brute-force analysis for interpreting permutations.

\subsection{Correctness of the test with reduced cases }

When $N$ is large, analysis can become challenging in the Approximate test due to time and space constraints. However, the Approximate test has the advantage of being able to reduce the test case size. To examine the effect of these reductions, we compared the full test results with those from tests reduced by factors of 2, 4, 8,...,256 where $N=32$. We analyzed the results from 8 bits to 16 bits using F-Y shuffling by Multiplication and Division with an AES-based RNG to observe the differences in analysis results when the test case was reduced, as shown in Figure~\ref{sample_result}, from (a) to (h). The full test comprises $246016(=32^2\times 31^2 / 4)$ cases.

The results show nearly identical trends from the test reduced by a factor of 2 to the test reduced by a factor of 256. The full-size test has already been presented in Figure~\ref{rng_fast}, and considering these results along with the overall aspect of reduced case tests, we can see that the observation of bias decreases at 10, 11, and 12 bits. However, we can still make a sufficient judgment about the base bit since the y-axis is on a log-2 scale. For instance, the 12 bits result in the test reduced by a factor of 256 is 8 times greater than the expected number of cases.

The Chi-square($\chi^2$) test results for each reduced case test are provided in Table~\ref{compare_reduce_table}. These results do not appear to differ substantially from the overall size test results of the reduced size test. As discussed in the previous section, the Chi-square($\chi^2$) test results do not account for a few outlier results, so we present the overall significant statistical processing. While there may be minor discrepancies between the full test results and those of the test reduced by a factor of 256, it's crucial to consider that the latter results were obtained by reducing the time complexity by a factor of 256. This means that a task that would have taken 256 days can now be completed in a day. Therefore, by collecting permutations uniformly and running multiple tests on the reduced cases to integrate the results, we can achieve meaningful outcomes.

\begin{figure*}
     \centering
     \begin{subfigure}[b]{0.40\textwidth}
         \centering
         \caption{}
         \includegraphics[width=\textwidth]{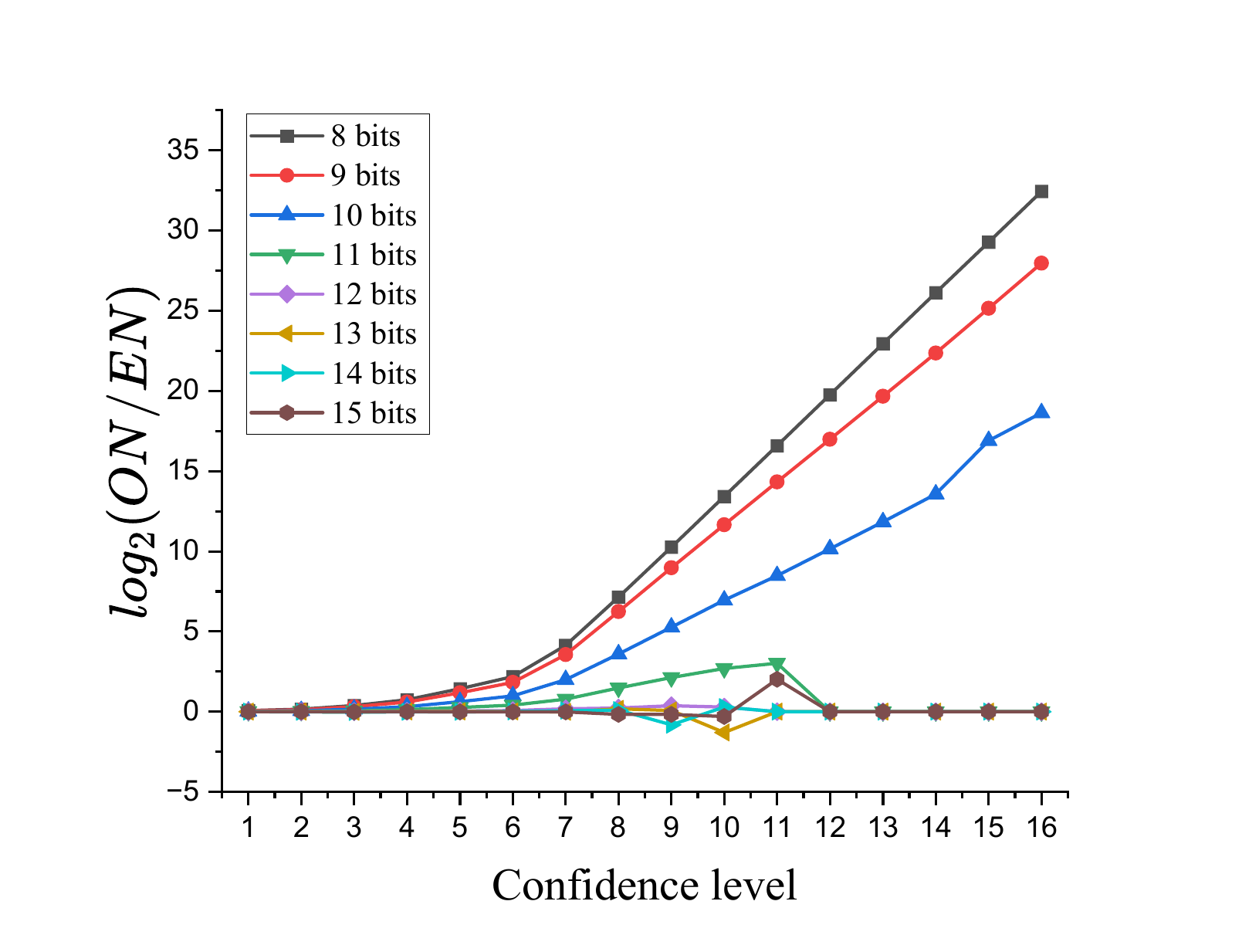}
         
         \label{fig:y equals x}
     \end{subfigure}
     \hfill
     \begin{subfigure}[b]{0.40\textwidth}
         \centering
         \caption{}
         \includegraphics[width=\textwidth]{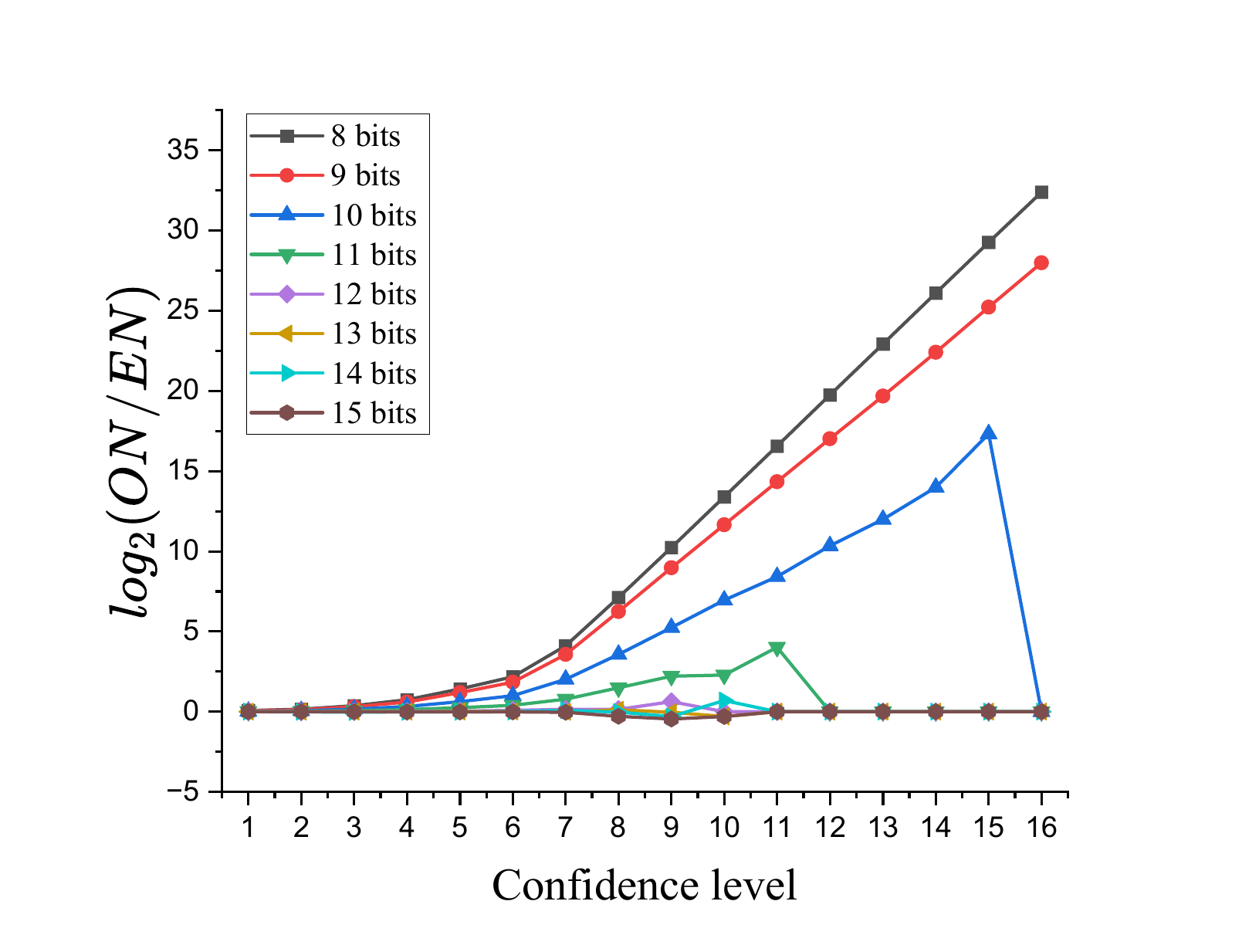}
         
         \label{fig:three sin x}
     \end{subfigure}
     \hfill
     \begin{subfigure}[b]{0.40\textwidth}
         \centering
         \caption{}
         \includegraphics[width=\textwidth]{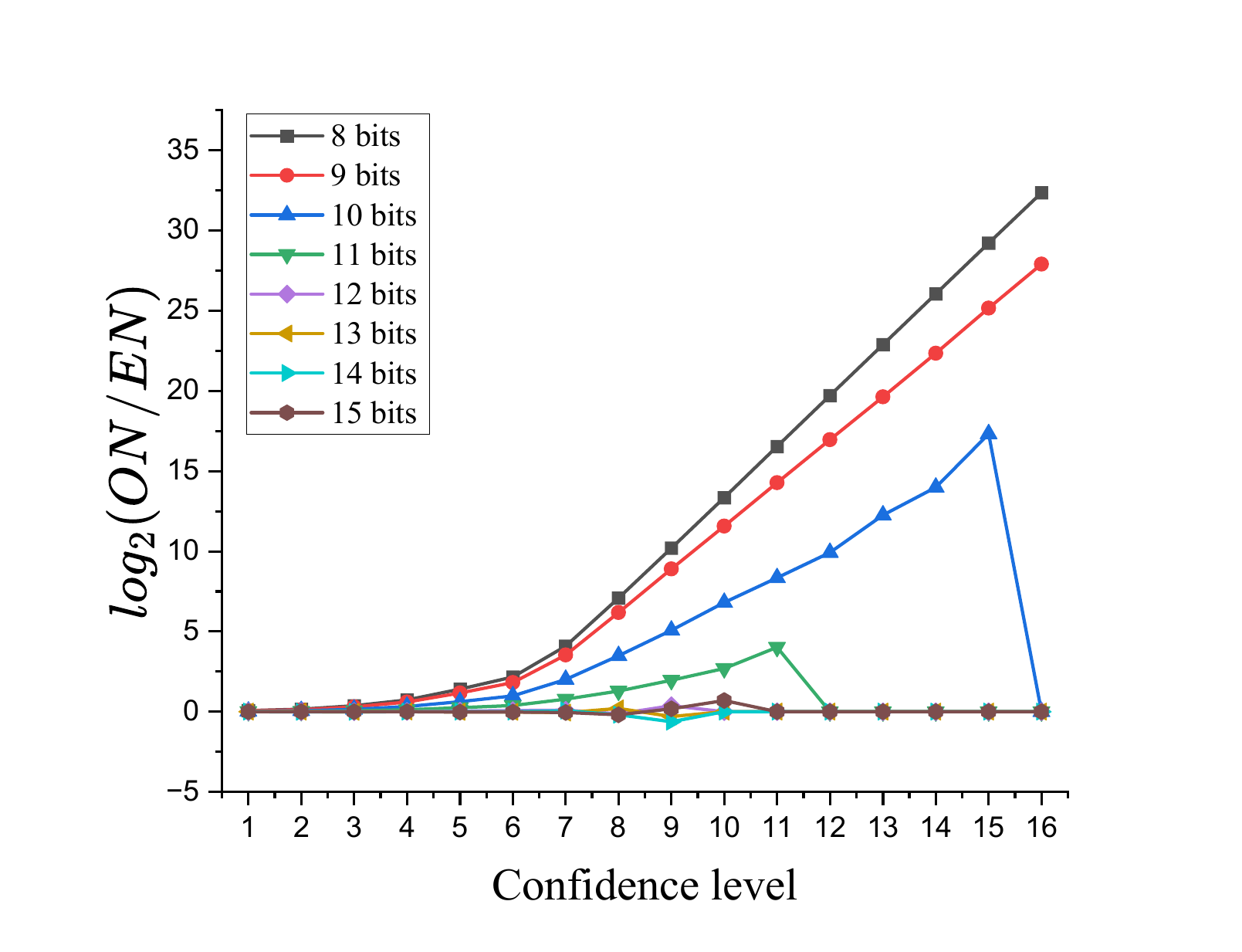}
         
         \label{fig:five over x}
     \end{subfigure}
     \hfill
     \begin{subfigure}[b]{0.40\textwidth}
         \centering
         \caption{}
         \includegraphics[width=\textwidth]{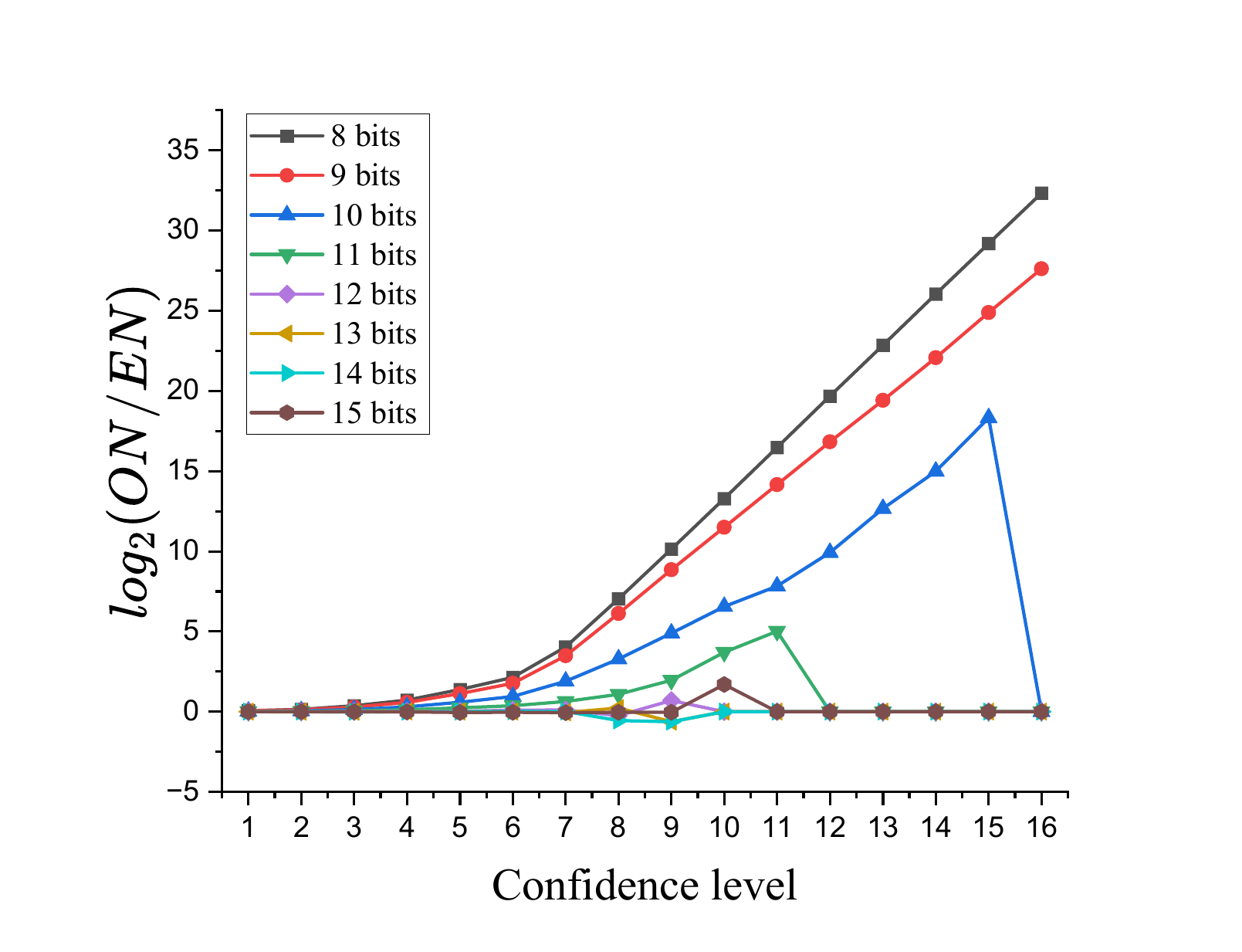}
         
         \label{fig:five over x}
     \end{subfigure}
          \hfill
     \begin{subfigure}[b]{0.40\textwidth}
         \centering
         \caption{}
         \includegraphics[width=\textwidth]{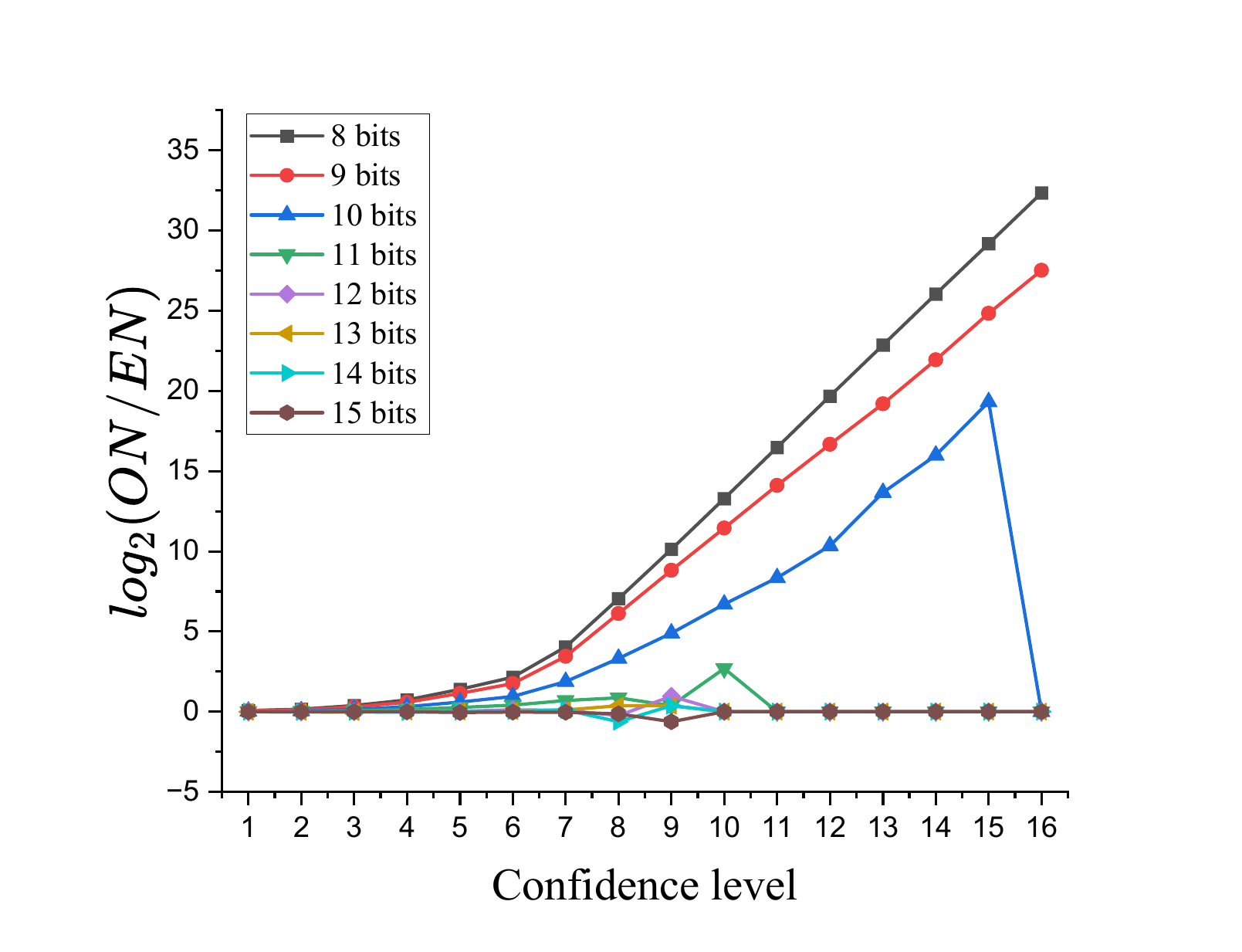}
         
         \label{fig:five over x}
     \end{subfigure}
     \hfill
     \begin{subfigure}[b]{0.40\textwidth}
         \centering
         \caption{}
         \includegraphics[width=\textwidth]{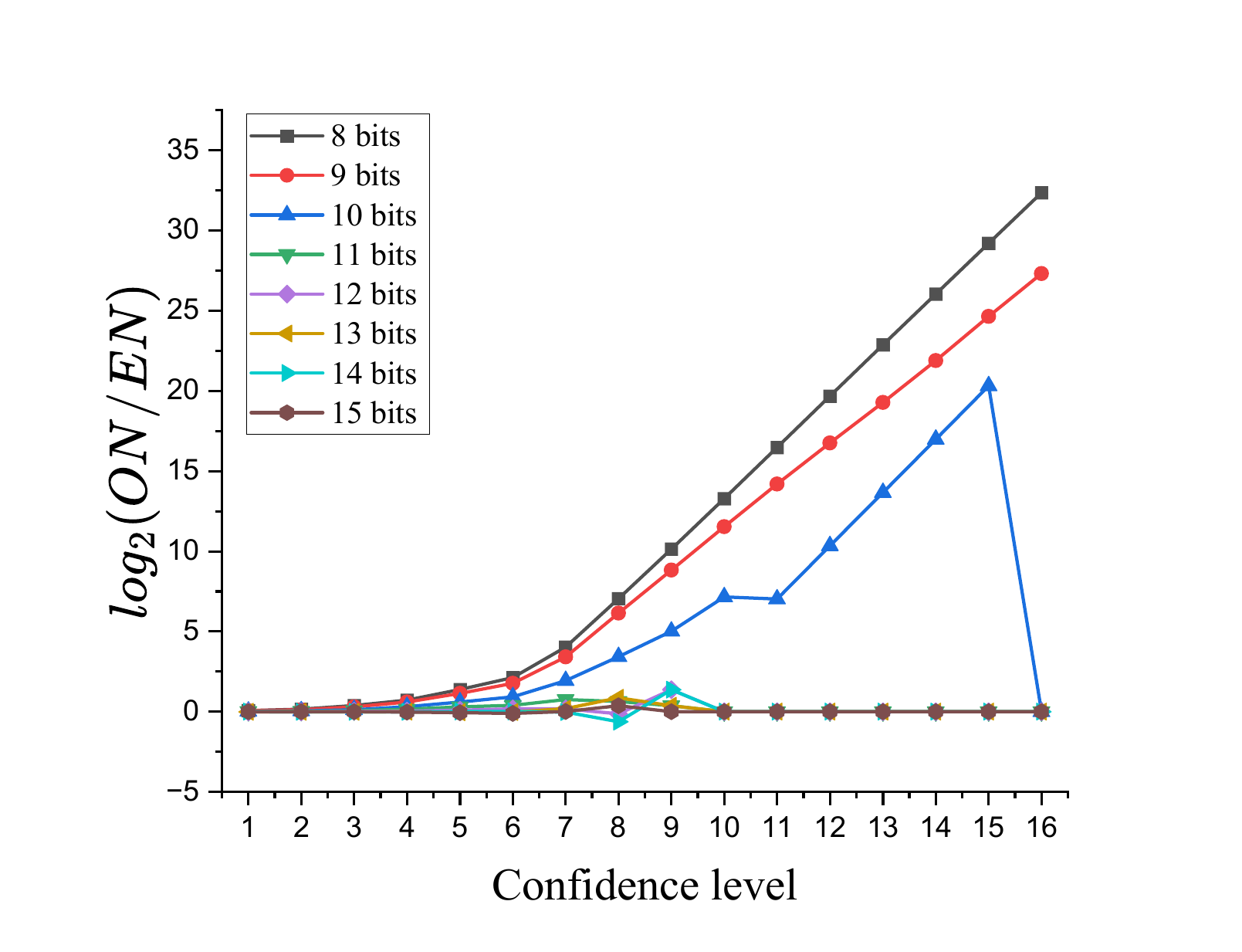}
         
         \label{fig:five over x}
     \end{subfigure}
     \hfill
     \begin{subfigure}[b]{0.40\textwidth}
         \centering
         \caption{}
         \includegraphics[width=\textwidth]{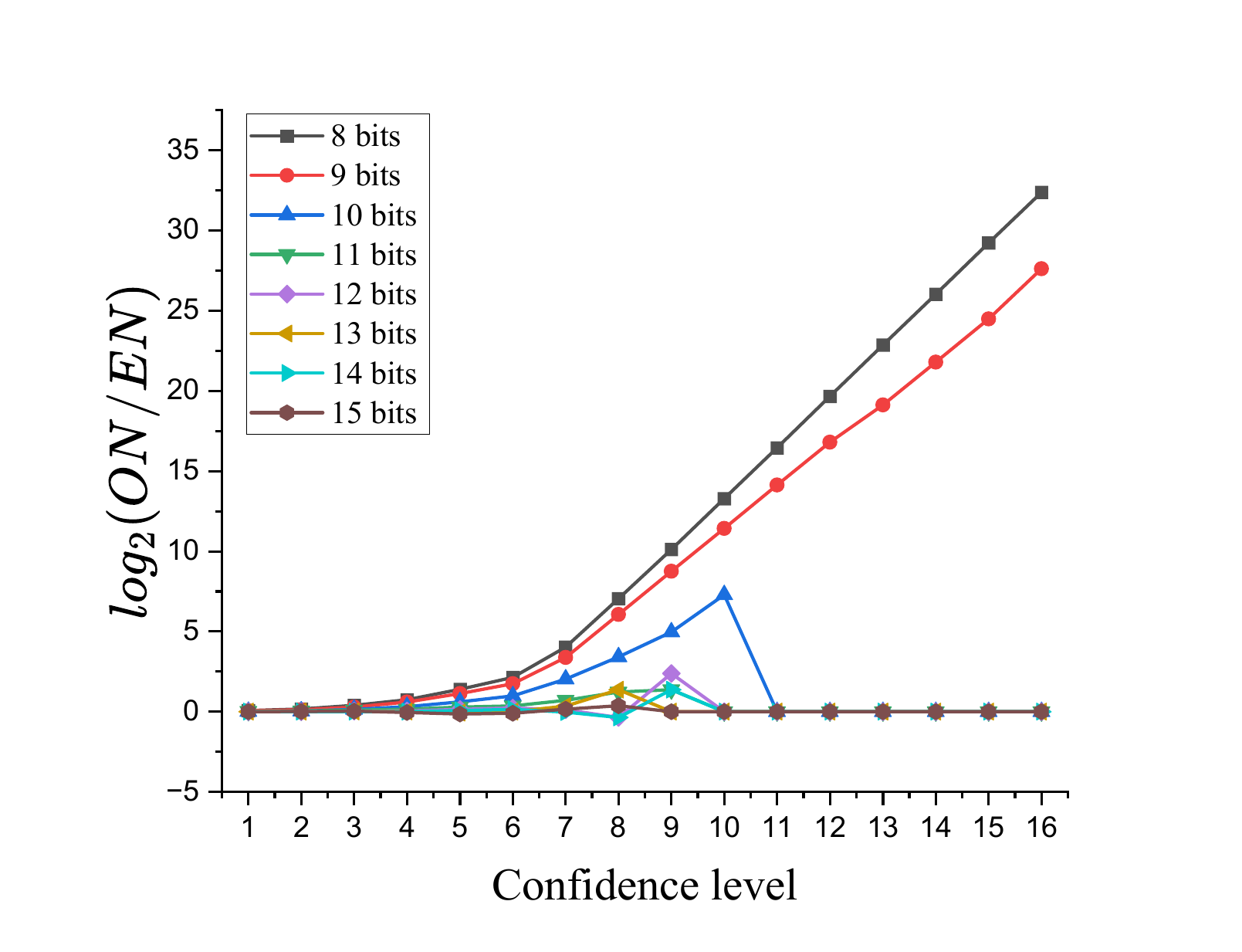}
         
         \label{fig:five over x}
     \end{subfigure}
     \hfill
     \begin{subfigure}[b]{0.40\textwidth}
         \centering
         \caption{}
         \includegraphics[width=\textwidth]{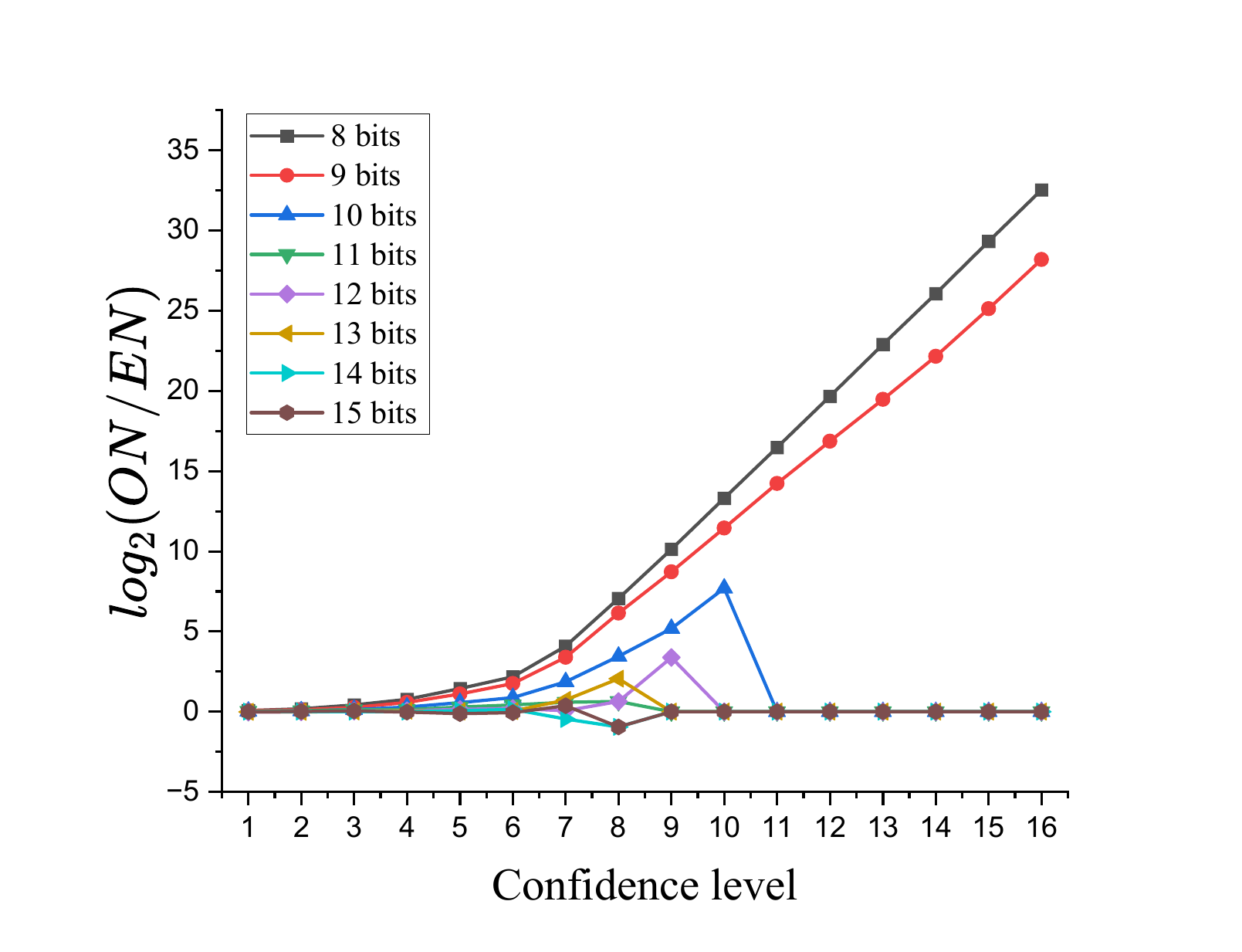}
         
         \label{fig:five over x}
     \end{subfigure}

        \caption{Results by reduced cases of test, (a)$\beta/2$, (b)$\beta/4$, (c)$\beta/8$, (d)$\beta/16$, (e)$\beta/32$, (f)$\beta/64$, (g)$\beta/128$, and (h)$\beta/256$ where $\beta$ is the number of cases; $N^2(N-1)^2/4$ ; y-axis is ratio of (ON/EN) and x-axis is the confidence level given by Table~\ref{xaxis}. Having a high value implies having a high bias (a bad outcome).}
        \label{sample_result}
        
\end{figure*}

\begin{table*}

\centering
\begin{tabular}{|p{25pt}|p{33pt}|p{33pt}|p{33pt}|p{33pt}|p{33pt}|p{33pt}|p{33pt}|p{33pt}|p{33pt}|}
\hline
\hline
RNG Bits & Full($\beta$) & $\beta/2$ & $\beta/4$ & $\beta/8$ & $\beta/16$ & $\beta/32$ & $\beta/64$ & $\beta/128$ & $\beta/256$ \\
\hline
8 & \underline{$<10^{-5}$}&   \underline{$<10^{-5}$}&   \underline{$<10^{-5}$}&   \underline{$<10^{-5}$}&   \underline{$<10^{-5}$}&   \underline{$<10^{-5}$}&   \underline{$<10^{-5}$}&   \underline{$<10^{-5}$}&   \underline{$<10^{-5}$}\\
9 & \underline{$<10^{-5}$}&   \underline{$<10^{-5}$}&   \underline{$<10^{-5}$}&   \underline{$<10^{-5}$}&   \underline{$<10^{-5}$}&   \underline{$<10^{-5}$}&   \underline{$<10^{-5}$}&   \underline{$<10^{-5}$}&   \underline{$<10^{-5}$}\\
10& \underline{$<10^{-5}$}&   \underline{$<10^{-5}$}&   \underline{$<10^{-5}$}&   \underline{$<10^{-5}$}&   \underline{$<10^{-5}$}&   \underline{$<10^{-5}$}&   \underline{$<10^{-5}$}&   \underline{$<10^{-5}$}&   \underline{$<10^{-5}$}\\
11& \underline{$<10^{-5}$}&   \underline{$<10^{-5}$}&   \underline{$<10^{-5}$}&   \underline{$<10^{-5}$}&   \underline{$<10^{-5}$}&   \underline{$<10^{-5}$}&   \underline{$<10^{-5}$}&   \underline{$<10^{-5}$}&  \underline{0.00003}\\
12& \underline{$<10^{-5}$}&   \underline{$<10^{-5}$}&   \underline{$<10^{-5}$}&  \underline{0.00005}&  \underline{0.00089}&  \underline{0.02429}&  \underline{0.00434}&  \underline{0.01253}&  0.07336\\
13& \underline{0.00062}&  \underline{0.00885}&  0.05110&  0.43300&  0.42282&  0.22490&  0.44989&  0.31136&  \underline{0.03796}\\
14&0.27722&  0.09129&  0.39604&  0.39302&  0.43122&  0.45117&  0.60578&  0.19814&  0.52007\\
15&0.58224&  0.29339&  0.68964&  0.63545&  0.70696&  0.61596&  0.70832&  0.63846&  0.40753\\
\hline
\end{tabular}
\caption{$1-$($p$-value) of Chi-square($\chi^2$) test by `reduced cases' by Brute-force test and Approximate $(N-1)$th order estimation. \\
(\underline{underline} number shows that $1-$($p$-value) is under $0.05$)}
\label{compare_reduce_table}

\end{table*}

\subsection{Discussion of the Experimental result}

According to Experimental result, we knows following main findings:   \\




\begin{itemize}
\item The Approximate $(N-1)$th order permutation estimator works effectively.
\end{itemize}

Our findings from the Approximate Test are as follows:

\begin{itemize}

\item \textbf{(Poor Random Source)} A poor random source significantly impacts the generation of random permutations with F-Y shuffling implementation. 
\item \textbf{(Application of Short Random Numbers)} We estimate that an additional 6 bits are needed to achieve a uniform distribution. For example, $N=32$, which requires 5 bits domain space, we should use at least an 11-bit random number for F-Y shuffling, whether we're using division or reduction.
\item \textbf{(Incorrect Implementation)} The NAIVE and Sattolo methods are not good choices, even with sufficient length and true random numbers. Particularly, the NAIVE method, which is widely used due to its fast speed and simple implementation, is not recommended due to its provable attack complexity. Other normal F-Y implementations generate uniform permutations. 

\end{itemize}

In conclusion, the incorrect implementations exemplified by the NAIVE and Sattolo methods should not be chosen, even when using true random numbers. Although the NAIVE method is popular for its speed and simplicity, it isn't advisable for use due to its vulnerability to attack.

\section{Conclusion}


In this research, we developed criteria for verifying perfect random permutation: namely, $(N-1)$th order estimation with normal distribution estimation. To establish this theory, we defined the $n$-th order permutation and a select function. We also proved that $(N-1)$th order permutation equals perfect random security. Additionally, based on the facts derived during the criteria-setting process, we found that higher order verification can replace lower order verification. For example, by verifying the 2nd order permutation, the 1st order permutation is naturally verified, saving verification time. The $(N-1)$th order permutation estimator is easy to implement, but it is not suitable for small sample sizes. Therefore, we provided an approximate definition and showed experimental results for commonly used permutation algorithms. We illustrated how three factors—incorrect implementation, poor random sources, and short random number applications—can lead to biased permutations. We demonstrated how each of these factors significantly influences the final graph of Approximate $(N-1)$th order permutations. Additionally, we provided simple guidelines for implementing F-Y shuffling.

In the era of Post-Quantum Cryptography (PQC), we anticipate that hiding countermeasures will be more widely utilized for PQC algorithms. Further research into hiding countermeasures and permutations is needed.

\noindent For future work, we plan to conduct the following studies:

\begin{itemize}
\item Investigate the practical effects of biased permutations on single and multi-trace attacks.
\item Develop a zero-redundancy, hardware-combined secure implementation for PQC without masking.
\item Create more efficient and perfect estimators for random permutations.
\end{itemize}



\begin{thebibliography}{00}

\bibitem{SOR94} Shor, Peter W. "Algorithms for quantum computation: discrete logarithms and factoring." In Proceedings 35th annual symposium on foundations of computer science, pp. 124-134. Ieee, 1994.

\bibitem{INDO_shuf} Pessl, Peter. "Analyzing the shuffling side-channel countermeasure for lattice-based signatures." In Progress in Cryptology–INDOCRYPT 2016: 17th International Conference on Cryptology in India, Kolkata, India, December 11-14, 2016, Proceedings 17, pp. 153-170. Springer International Publishing, 2016.

\bibitem{DPA} Kocher, Paul, Joshua Jaffe, and Benjamin Jun. "Differential power analysis." In Advances in Cryptology—CRYPTO’99: 19th Annual International Cryptology Conference Santa Barbara, California, USA, August 15–19, 1999 Proceedings 19, pp. 388-397. Springer Berlin Heidelberg, 1999.

\bibitem{timming} Kocher, Paul C. "Timing attacks on implementations of Diffie-Hellman, RSA, DSS, and other systems." In Advances in Cryptology—CRYPTO’96: 16th Annual International Cryptology Conference Santa Barbara, California, USA August 18–22, 1996 Proceedings 16, pp. 104-113. Springer Berlin Heidelberg, 1996.

\bibitem{CPA} Brier, Eric, Christophe Clavier, and Francis Olivier. "Correlation power analysis with a leakage model." In Cryptographic Hardware and Embedded Systems-CHES 2004: 6th International Workshop Cambridge, MA, USA, August 11-13, 2004. Proceedings 6, pp. 16-29. Springer Berlin Heidelberg, 2004.

\bibitem{template} Chari, Suresh, Josyula R. Rao, and Pankaj Rohatgi. "Template attacks." In Cryptographic Hardware and Embedded Systems-CHES 2002: 4th International Workshop Redwood Shores, CA, USA, August 13–15, 2002 Revised Papers 4, pp. 13-28. Springer Berlin Heidelberg, 2003.

\bibitem{SPA} Mangard, Stefan. "A simple power-analysis (SPA) attack on implementations of the AES key expansion." In Information Security and Cryptology—ICISC 2002: 5th International Conference Seoul, Korea, November 28–29, 2002 Revised Papers 5, pp. 343-358. Springer Berlin Heidelberg, 2003.

\bibitem{collision} Bogdanov, Andrey, Ilya Kizhvatov, and Andrey Pyshkin. "Algebraic methods in side-channel collision attacks and practical collision detection." In Progress in Cryptology-INDOCRYPT 2008: 9th International Conference on Cryptology in India, Kharagpur, India, December 14-17, 2008. Proceedings 9, pp. 251-265. Springer Berlin Heidelberg, 2008.

\bibitem{ASCA} Renauld, Mathieu, and François-Xavier Standaert. "Algebraic side-channel attacks." In International Conference on Information Security and Cryptology, pp. 393-410. Berlin, Heidelberg: Springer Berlin Heidelberg, 2009.

\bibitem{single_pqc} Pessl, Peter, and Robert Primas. "More practical single-trace attacks on the number theoretic transform." In Progress in Cryptology–LATINCRYPT 2019: 6th International Conference on Cryptology and Information Security in Latin America, Santiago de Chile, Chile, October 2–4, 2019, Proceedings 6, pp. 130-149. Springer International Publishing, 2019.

\bibitem{single_KMU} Sim, Bo-Yeon, Jihoon Kwon, Joohee Lee, Il-Ju Kim, Tae-Ho Lee, Jaeseung Han, Hyojin Yoon, Jihoon Cho, and Dong-Guk Han. "Single-trace attacks on message encoding in lattice-based KEMs." IEEE Access 8 (2020): 183175-183191.

\bibitem{SASCA} Veyrat-Charvillon, Nicolas, Benoît Gérard, and François-Xavier Standaert. "Soft analytical side-channel attacks." In Advances in Cryptology–ASIACRYPT 2014: 20th International Conference on the Theory and Application of Cryptology and Information Security, Kaoshiung, Taiwan, ROC, December 7-11, 2014. Proceedings, Part I 20, pp. 282-296. Springer Berlin Heidelberg, 2014.

\bibitem{book} Mangard, Stefan, Elisabeth Oswald, and Thomas Popp. Power analysis attacks: Revealing the secrets of smart cards. Vol. 31. Springer Science \& Business Media, 2008.

\bibitem{k_trace}  Hamburg, Mike, Julius Hermelink, Robert Primas, Simona Samardjiska, Thomas Schamberger, Silvan Streit, Emanuele Strieder, and Christine van Vredendaal. "Chosen ciphertext k-trace attacks on masked cca2 secure kyber." IACR Transactions on Cryptographic Hardware and Embedded Systems (2021): 88-113.

\bibitem{Shu_chi} Mitchell, Rory, Daniel Stokes, Eibe Frank, and Geoffrey Holmes. "Bandwidth-optimal random shuffling for GPUs." ACM Transactions on Parallel Computing 9, no. 1 (2022): 1-20.


\bibitem{shuffling} Veyrat-Charvillon, Nicolas, Marcel Medwed, Stéphanie Kerckhof, and François-Xavier Standaert. "Shuffling against side-channel attacks: A comprehensive study with cautionary note." In Advances in Cryptology–ASIACRYPT 2012: 18th International Conference on the Theory and Application of Cryptology and Information Security, Beijing, China, December 2-6, 2012. Proceedings 18, pp. 740-757. Springer Berlin Heidelberg, 2012.

\bibitem{SE_PQC}  Bettale, Luk, Simon Montoya, and Guénaël Renault. "Safe-error analysis of post-quantum cryptography mechanisms-short paper." In 2021 Workshop on Fault Detection and Tolerance in Cryptography (FDTC), pp. 39-44. IEEE, 2021.

\bibitem{SE_O} Yen, Sung-Ming, and Marc Joye. "Checking before output may not be enough against fault-based cryptanalysis." IEEE Transactions on computers 49, no. 9 (2000): 967-970.

\bibitem{NIST20} NIST, PQC Selected Algorithms [Online] Available :  https://csrc.nist.gov/Projects/post-quantum-cryptography/selected-algorithms-2022. Accessed on: Otc 2022. 

\bibitem{NIST22} NIST, NIST round 4 submission, [Online] Available : https://csrc.nist.gov/Projects/post-quantum-cryptography/round-4-submissions. Accessed on: Sep 2022. 

\bibitem{BP_PQC} Hermelink, Julius, Silvan Streit, Emanuele Strieder, and Katharina Thieme. "Adapting belief propagation to counter shuffling of NTTs." IACR Transactions on Cryptographic Hardware and Embedded Systems (2023): 60-88.

\bibitem{bellcore} Boneh, Dan, Richard A. DeMillo, and Richard J. Lipton. "On the importance of checking cryptographic protocols for faults." In International conference on the theory and applications of cryptographic techniques, pp. 37-51. Berlin, Heidelberg: Springer Berlin Heidelberg, 1997.

\bibitem{mix_and_cut} Ristenpart, Thomas, and Scott Yilek. "The mix-and-cut shuffle: small-domain encryption secure against N queries." In Advances in Cryptology–CRYPTO 2013: 33rd Annual Cryptology Conference, Santa Barbara, CA, USA, August 18-22, 2013. Proceedings, Part I, pp. 392-409. Springer Berlin Heidelberg, 2013.

\bibitem{FPE} Bellare, Mihir, Thomas Ristenpart, Phillip Rogaway, and Till Stegers. "Format-preserving encryption." In Selected Areas in Cryptography: 16th Annual International Workshop, SAC 2009, Calgary, Alberta, Canada, August 13-14, 2009, Revised Selected Papers 16, pp. 295-312. Springer Berlin Heidelberg, 2009.

\bibitem{FPE2} Black, John, and Phillip Rogaway. "Ciphers with arbitrary finite domains." In Topics in Cryptology—CT-RSA 2002: The Cryptographers’ Track at the RSA Conference 2002 San Jose, CA, USA, February 18–22, 2002 Proceedings, pp. 114-130. Springer Berlin Heidelberg, 2002.

\bibitem{Knuth} Knuth, Donald E. Art of computer programming, volume 2: Seminumerical algorithms. Addison-Wesley Professional, 2014.

\bibitem{FY} Fisher, Ronald Aylmer, and Frank Yates. Statistical tables for biological, agricultural, and medical research. Hafner Publishing Company, 1953.

\bibitem{per_a1} Prodinger, Helmut. "On the analysis of an algorithm to generate a random cyclic permutation." Ars Combinatoria 65 (2002): 75-78.

\bibitem{per_a2} Mahmoud, Hosam M. "Mixed distributions in Sattolo's algorithm for cyclic permutations via randomization and derandomization." Journal of applied probability 40, no. 3 (2003): 790-796.

\bibitem{MSDN} Microsoft learn, \\https://learn.microsoft.com/en-us/previous-versions/398ax69y(v=vs.140) 

\bibitem{Phd_T}Borga, Jacopo. "Random Permutations--A geometric point of view." arXiv preprint arXiv:2107.09699 (2021).

\bibitem{ref_per1}V. Gonˇcarov. On the field of combinatory analysis. Amer. Math. Soc. Transl. (2), 19:1–46, 1962.

\bibitem{ref_per2}Diaconis, Persi, and Ronald L. Graham. "Spearman's footrule as a measure of disarray." Journal of the Royal Statistical Society Series B: Statistical Methodology 39, no. 2 (1977): 262-268.

\bibitem{ref_per3}Baik, Jinho, Percy Deift, and Kurt Johansson. "On the distribution of the length of the longest increasing subsequence of random permutations." Journal of the American Mathematical Society 12, no. 4 (1999): 1119-1178.

\bibitem{ref_per4}Fulman, Jason. "Stein's method and non-reversible Markov chains." Lecture Notes-Monograph Series (2004): 69-77.

\bibitem{ref_per5} Janson, Svante, Brian Nakamura, and Doron Zeilberger. "On the asymptotic statistics of the number of occurrences of multiple permutation patterns." arXiv preprint arXiv:1312.3955 (2013).

\bibitem{timing_HQC} Guo, Qian, Clemens Hlauschek, Thomas Johansson, Norman Lahr, Alexander Nilsson, and Robin Leander Schröder. "Don’t reject this: Key-recovery timing attacks due to rejection-sampling in HQC and BIKE." IACR Transactions on Cryptographic Hardware and Embedded Systems (2022): 223-263.


\bibitem{master} Vennos, Amy Demetra Geae. "Security of Lightweight Cryptographic Primitives." PhD diss., Virginia Tech, 2021.

\bibitem{secure_BIKE} Sendrier, Nicolas. "Secure sampling of constant-weight words–application to bike." Cryptology ePrint Archive (2021).

\bibitem{ACNS2007} Tillich, Stefan, Christoph Herbst, and Stefan Mangard. "Protecting AES software implementations on 32-bit processors against power analysis." In Applied Cryptography and Network Security: 5th International Conference, ACNS 2007, Zhuhai, China, June 5-8, 2007. Proceedings 5, pp. 141-157. Springer Berlin Heidelberg, 2007.






\end{thebibliography}

\appendix

\section{The proof of Proposition}\label{proof_proposition}
Here is the proof of the Proposition~\ref{proposition_main}. \newline


\noindent\textbf{Proposition~\ref{proposition_main}}. Let the probability \( P(A|B_1) = P(A|B_2) = \dots = P(A|B_n) = k \) and \( B_i \cap B_j = \phi \) for all \( i \neq j \). Then, \( P(A|\bigcup_{i=1}^{n}B_i) = k \).

\begin{proof}
By the definition of conditional probability, we have
\[ P(A|\bigcup_{i=1}^{n}B_i) = \frac{P(A \cap \bigcup_{i=1}^{n}B_i)}{P(\bigcup_{i=1}^{n}B_i)} \]

Using the assumption that \( B_i \cap B_j = \phi \) for all \( i \neq j \), this becomes
\[ P(A|\bigcup_{i=1}^{n}B_i) = \frac{\bigcup_{i=1}^{n} P(A \cap B_i)}{\sum_{i=1}^{n} P(B_i)} \]

Given that \( P(A \cap B_i) = k \times P(B_i) \), this simplifies to
\[ P(A|\bigcup_{i=1}^{n}B_i) = k \times \frac{\sum_{i=1}^{n} P(B_i)}{\sum_{i=1}^{n} P(B_i)} = k \]
\end{proof}

\section{The proof of the Remark}
\subsection{The proof of Remark}\label{proof_theorem}
Here is the proof of the Remark~\ref{main_theorem}. We need a lemma to prove it. 

\begin{lemma}\label{lemma1}
If \( (B_i \cap B_j) = \phi \) for all \( i \neq j \) and \( P(\bigcup_{i} B_i | A) = 1 \), then \( \sum_{i} P(A \cap B_i) = P(A) \).
\end{lemma}

\begin{proof}
Firstly, observe that \( (A \cap B_i) \cap (A \cap B_j) = \phi \) for all \( i \neq j \) due to the disjoint nature of \( B_i \) sets. Hence, 
\[ \sum_{i} P(A \cap B_i) = P\left(\bigcup_{i} (A \cap B_i)\right) \]

Using the properties of conditional probability, we get
\[ P(A \cap \bigcup_{i} B_i) = P(\bigcup_{i} B_i | A) \times P(A) \]

Given that \( P(\bigcup_{i} B_i | A) = 1 \), it follows that
\[ \sum_{i} P(A \cap B_i) = P(A) \]
\end{proof}

\noindent\textbf{Remark~\ref{main_theorem}}. \textit{A select function $\gamma$ has the $k$-th order permutation for all $k\leq N$, then $\gamma$ provides $(k-1)$th order permutation. }

\begin{proof}
Assume that \(\gamma\) provides a \(k\)-th order permutation. According to definition~\ref{main_definition}, given a randomly selected \(t\) such that \(\gamma(t) = f_m\), the probability is
\[ P(f_m(x) = y | \bigcap_{j=1}^{k-1} f_m(a_j) = b_j) = \frac{1}{N-k+1} \]
for a specified \(m \in \mathbb{Z}_{N!}\) and for all \(x, y, a_i, b_i \in \mathbb{Z_N}\) with \(x \neq a_i\) and \(y \neq b_i\).

Our goal is to show that
\[ P(f_m(x) = y | \bigcap_{i=1}^{k-2} f_m(a_i) = b_i) = \frac{1}{N-k+2} \]
under analogous constraints for \(m, x, y, a_i, b_i\).

For simplicity, denote
\[ P(f_m(x) = y | \bigcap_{i=1}^{k-1} f_m(a_i) = b_i) = P(A | C_1 C_2 \dots C_{k-1}) \]
where \(A\) is any event satisfying the definition. Since \(\gamma\) ensures a \(k\)-th order permutation, the conditional probabilities equate as:
\[ P(A | C_1 C_2 \dots C_{k-1}) = P(C_1 | A C_2 \dots C_{k-1}) \]

Expressing this in terms of intersections:
\[ \frac{P(A \cap C_1 \cap \dots \cap C_{k-1})}{P(C_1 \cap \dots \cap C_{k-1})} = \frac{P(A \cap C_1 \cap \dots \cap C_{k-1})}{P(A \cap C_2 \cap \dots \cap C_{k-1})} \]

From which, we deduce that
\[ P(C_1 \cap \dots \cap C_{k-1}) = P(A \cap C_2 \cap \dots \cap C_{k-1}) \]
implying that each term in the intersection \(P(\bigcap_{i=1}^{k-1} f_m(a_i) = b_i)\) has equivalent probability.

Considering \(k-2\) fixed conditions and varying one, define \(\delta = \bigcap_{i=1}^{k-2} f_m(a_i) = b_i\). In this case, \(f_m(x')\) can take on \(N-k+2\) distinct values. Enumerating these values by \(f(x') = y_{e_l}\) and noting their mutual exclusivity (by Lemma~\ref{lemma1}), we have:
\[ \sum_{l=1}^{N-k+2} P(f_m(x') = y_{e_l} \cap \delta) = P(\delta) \]

Given the equivalence of terms \(P(f_m(x') = y_e \cap \delta)\), we deduce:
\[ P(f_m(x') = y_{e_1} \cap \delta) = \frac{P(\delta)}{N-k+2} \]
and hence,
\[ \frac{P(f_m(x') = y_{e_1} \cap \delta)}{P(\delta)} = \frac{1}{N-k+2} \]

Thus, we arrive at:
\[ \frac{P(f_m(x') = y_{k-1} \cap \bigcap_{i=1}^{k-2} f_m(a_i) = b_i)}{P(\bigcap_{i=1}^{k-2} f_m(a_i) = b_i)} = \frac{1}{N-k+2} \]

Consequently, by definition~\ref{main_definition}, \(\gamma\) offers a \((k-1)\)-th order permutation.
\end{proof}

\section{An example for the second order permutation but not the third order permutation with five-indices}\label{App:Example_2nd}
\begin{center}
\noindent\{1,2,3,4,5\},
\{1,2,3,5,4\},
\{1,3,2,4,5\},
\{1,3,2,5,4\},\\
\{1,4,5,2,3\},
\{1,4,5,3,2\},
\{1,5,4,2,3\},
\{1,5,4,3,2\},\\
\{2,1,4,3,5\},
\{2,1,4,5,3\},
\{2,3,5,1,4\},
\{2,3,5,4,1\},\\
\{2,4,1,3,5\},
\{2,4,1,5,3\},
\{2,5,3,1,4\},
\{2,5,3,4,1\},\\
\{3,1,5,2,4\},
\{3,1,5,4,2\},
\{3,2,4,1,5\},
\{3,2,4,5,1\},\\
\{3,4,2,1,5\},
\{3,4,2,5,1\},
\{3,5,1,2,4\},
\{3,5,1,4,2\},\\
\{4,1,3,2,5\},
\{4,1,3,5,2\},
\{4,2,5,1,3\},
\{4,2,5,3,1\},\\
\{4,3,1,2,5\},
\{4,3,1,5,2\},
\{4,5,2,1,3\},
\{4,5,2,3,1\},\\
\{5,1,2,3,4\},
\{5,1,2,4,3\},
\{5,2,1,3,4\},
\{5,2,1,4,3\},\\
\{5,3,4,1,2\},
\{5,3,4,2,1\},
\{5,4,3,1,2\},
\{5,4,3,2,1\}. 
\end{center}

\end{document}